%% file: main.tex
\documentclass[letterpaper, 10 pt, conference]{ieeeconf}  % Comment this line out
\makeatletter
\let\NAT@parse\undefined
\makeatother

\IEEEoverridecommandlockouts                              % This command is only
\input{auxFiles/packages}

\title{\LARGE \bf
Physics-Informed Learning of Feedback-Linearizing Representations}

\author{Pavlos Kallinikidis$^{*1}$, Fengjun Yang$^{*1}$, David Snyder$^{1}$, Jacob H. Seidman$^{2}$, \\Nikolai Matni$^{1}$, Paris Perdikaris$^{3}$, and George J. Pappas$^{1}$%
\thanks{$^*$Equal contribution. $^{1}$GRASP Lab, University of Pennsylvania, Philadelphia, PA, USA. $^{2}$Reality Defender Inc. $^{3}$Department of Mechanical Engineering and Applied Mechanics, University of Pennsylvania, Philadelphia, PA, USA. This work was supported in part by DARPA TRS program under contract HR00112590145.}%
}

\begin{document}

\setlength{\belowcaptionskip}{-10pt}

\maketitle
\thispagestyle{empty}
\pagestyle{empty}

%%%%%%%%%%%%%%%%%%%%%%%%%%%%%%%%%%%%%%%%%%%%%%%%%%%%%%%%%%%%%%%%%%%%%%%%%%%%%%%%

\input{sections/00_abstract}
\input{sections/1-introduction}

\input{sections/2-problem-statement}

\input{sections/3-fbl}
\input{sections/4-learning}
\input{sections/5-controller}
\input{sections/6-experiments}

\input{sections/7-conclusion}

% \input{sections/simple_experiments}
% Empty placeholder omitted for arXiv.

% References are important to the reader; therefore, each citation must be complete and correct. If at all possible, references should be commonly available publications.

\bibliographystyle{bibFiles/IEEEbib}
\bibliography{bibFiles/main}

\clearpage
\onecolumn
\appendix
\input{sections/A1-MIMO}
\input{sections/A2-error-bounds}
\input{sections/A3-on-the-diffeomorphism}
\input{sections/A4-synthesis}

\end{document}

%% file: auxFiles/packages.tex
\usepackage{amsmath, amssymb, amsthm} % Standard AMS packages
\usepackage{bbm}                      % Blackboard bold 1 (\mathbbm{1})
\usepackage{cancel}

\usepackage{graphicx}     % Include graphics
\usepackage{epstopdf}     % EPS to PDF conversion
\usepackage{wrapfig}      % Wrap text around figures
\usepackage{float}        % More precise float placement
\usepackage{caption}      % Custom captions
\usepackage{subcaption}   % Subfigures
\usepackage{tikz}         % Drawing figures
\usepackage{pgfplots}     % Plots with TikZ
\usepackage{booktabs}

\usepackage{listings}                                % Code listings
\usepackage[ruled,vlined,linesnumbered]{algorithm2e} % Algorithms

\usepackage{xcolor}           % Colors
\usepackage{textcomp}         % Extra text symbols
\usepackage[utf8]{inputenc} % allow utf-8 input
\usepackage{csquotes}         % Quotation handling
\usepackage[most]{tcolorbox}  % Colored boxes (e.g. theorems, notes)
\usepackage{url}
\usepackage{mathdots}
\usepackage{algorithmic}
\usepackage{fullpage}      % Use full page
\usepackage{etoolbox}      % Programming tools for LaTeX
\let\labelindent\relax
\usepackage{enumitem}      % Customizable lists
\usepackage{makecell}      % Line breaks inside table cells
\usepackage{authblk}       % Author affiliations
\usepackage{balance}       % Balance final page columns
\usepackage{multirow}      % Multi-row table

\usepackage{todonotes}     % Todo notes

\usepackage{auxFiles/symbolDef}
\usepackage{auxFiles/commands}
\input{auxFiles/auxCommands}

\usepackage[T1]{fontenc}    % use 8-bit T1 fonts
\usepackage{booktabs}       % professional-quality tables
\usepackage{nicefrac}       % compact symbols for 1/2, etc.
\usepackage{microtype}      % microtypography
\usepackage[normalem]{ulem}

%% file: auxFiles/auxCommands.tex
\makeatletter
    \@ifpackageloaded{xcolor}{}{\usepackage{xcolor}}
\makeatother

\newcommand{\blue}  [1] {{\color{blue} #1}}

\makeatletter
    \@ifpackageloaded{needspace}{}{\usepackage{needspace}}
\makeatother

\newcommand{\uppercaseGreek}[1]{%
    \begingroup
    \ucmathlist\MakeUppercase{#1}%
    \endgroup
}

\newcommand{\ucmathlist}{%
    \def\alpha{\mathrm{A}}%
    \def\beta{\mathrm{B}}%
    \let\gamma=\Gamma
    \let\delta=\Delta
    \def\epsilon{\mathrm{E}}%
    \def\varepsilon{\mathrm{E}}%
    \def\zeta{\mathrm{Z}}%
    \def\eta{\mathrm{H}}%
    \let\theta=\Theta
    \let\vartheta=\Theta
    \def\iota{\mathrm{I}}%
    \def\kappa{\mathrm{K}}%
    \let\lambda=\Lambda
    \def\mu{\mathrm{M}}%
    \def\nu{\mathrm{N}}%
    \let\xi=\Xi
    \let\pi=\Pi
    \let\varpi=\Pi
    \def\rho{\mathrm{P}}%
    \def\varrho{\mathrm{P}}%
    \let\sigma=\Sigma
    \def\tau{\mathrm{T}}%
    \let\upsilon=\Upsilon
    \let\phi=\Phi
    \let\varphi=\Phi
    \def\chi{\mathrm{X}}%
    \let\psi=\Psi
    \let\omega=\Omega
}

\makeatletter
    \@ifpackageloaded{amsthm}{}{

        \usepackage{amsthm}
    }
\makeatother
\theoremstyle{plain}
    \newtheorem{theorem}{Theorem}
    
    \newtheorem{lemma}[theorem]{Lemma}
    
\theoremstyle{definition}

    \newtheorem{assumption}{Assumption}

\makeatletter
\def\renewtheorem#1{%
    \expandafter\let\csname#1\endcsname\relax
    \expandafter\let\csname c@#1\endcsname\relax
    \gdef\renewtheorem@envname{#1}
    \renewtheorem@secpar
}
\def\renewtheorem@secpar{\@ifnextchar[{\renewtheorem@numberedlike}{\renewtheorem@nonumberedlike}}
\def\renewtheorem@numberedlike[#1]#2{\newtheorem{\renewtheorem@envname}[#1]{#2}}
\def\renewtheorem@nonumberedlike#1{  
    \def\renewtheorem@caption{#1}
    \edef\renewtheorem@nowithin{\noexpand\newtheorem{\renewtheorem@envname}{\renewtheorem@caption}}
    \renewtheorem@thirdpar
}
\def\renewtheorem@thirdpar{\@ifnextchar[{\renewtheorem@within}{\renewtheorem@nowithin}}
\def\renewtheorem@within[#1]{\renewtheorem@nowithin[#1]}
\makeatother

%% file: sections/00_abstract.tex
\begin{abstract}
Feedback linearization is a powerful tool in nonlinear control, but finding the linearizing coordinate transform remains challenging. Feedback linearizing transforms are governed by well-established partial differential equations (PDEs) whose well-posedness is based on the Lie-algebraic Frobenius theorem, but solving the PDEs becomes computationally intractable for large system dimensions. To address this challenge, we propose a cascaded physics-informed neural network (PINN) framework to approximately solve these PDEs. By separately parameterizing terms of different Lie derivative orders, we mitigate the compounding errors inherent in fitting high-order Lie derivatives of neural networks. Furthermore, we use the learned transform to design a tracking controller and establish theoretical bounds on its error with respect to inaccuracies in the learned feedback-linearizing representation. We validate our method on a broad class of feedback linearizable systems, including a multi-input, multi-output planar quadrotor and synthetic examples for which analytical approaches are impractical, and demonstrate that our approach can computationally discover effective feedback-linearizing representations of nonlinear systems for control tasks.

% Feedback linearization is a powerful tool in nonlinear control, but finding the linearizing coordinate transform is analytically complex and commonly requires bespoke analysis for each system considered. Learning-based control methods, by contrast, provide the capacity to adapt to system structure, but purely data-driven methods often ignore the rich geometric structure of the problem. Building on recent advances in algorithms that can compute high-order Lie derivatives efficiently, we present a method for learning feedback-linearizing representations by explicitly leveraging the inherent geometric structure of feedback linearizable systems. Without assuming \textit{a priori} knowledge of the linearizing output, our approach bounds the downstream system regulation error with respect to errors in the learned representation. We validate our method on a broad class of dynamical systems, including a multi-input, multi-output planar quadrotor and several systems for which analytical approaches are impractical; we demonstrate that our approach can recover representations that exactly feedback linearize a nonlinear system. 
\end{abstract}

%% file: sections/1-introduction.tex
 \section{Introduction}
\label{sec:introduction}
Feedback linearization transforms nonlinear dynamics into an equivalent linear system via feedback transformation~\cite{isidori1985nonlinear}. Once found, this transform enables the application of powerful tools from linear systems theory to control complex, nonlinear systems. The existence of such transforms depends on a Lie-algebraic Frobenius condition; when the Frobenius condition holds, there exists a well-posed system of partial differential equations (PDEs) whose solution and associated Lie derivatives provide the linearizing coordinates.
% Although well-known PDE conditions exist to characterize such transformations via a set of \emph{linearizing outputs} \cite{isidori1985nonlinear},
However, solving the PDEs is challenging in general. Analytical approaches are customized and highly system-specific, while algorithmic approaches are scarce and suffer from significant computational shortcomings \citet{Tall2009ExplicitFL}.

Recently, PINNs \cite{raissi2019physics} have emerged as an efficient method to approximate PDE solutions with applications to, for example, optimal control \cite{fotiadis2023physics} and reachability \cite{bansal2021deepreach}. PINNs recast PDE solutions as minimizers of a residual loss. As such, their computational cost empirically scales with solution complexity rather than grid resolution \cite{bansal2021deepreach}, improving performance on higher-dimensional problems.

We leverage the computational benefits of PINNs to systematically discover feedback-linearizing representations of nonlinear systems. Two key technical challenges arise. First, during learning, the high-order Lie derivatives in the governing PDEs introduce computational bottlenecks and compounding approximation errors. To overcome this, we introduce a cascaded neural network structure that
% exploits the geometric structure of feedback linearizable systems and
decomposes complex, high-order PDEs into a sequence of tractable, first-order residual minimizations. We demonstrate that this approach significantly outperforms naive PINN baselines and accurately reconstructs linearizing behavior. Second, controller design must account for the effects of a learned (approximate) representation on downstream control performance. We design a controller using the learned transform and bound its downstream tracking error with respect to learning errors. Evaluated on both physical and synthetic systems up to eight dimensions, our work constitutes the first systematic application of PINNs for discovering feedback linearizing representations.

\subsection{Related Work}
\label{sec:related_work}

\subsubsection{Learning Linear Representations of Nonlinear Systems}
Prior learning-based approaches for feedback linearization focus primarily on compensating for uncertain dynamics when the linearizing outputs are already known, whether by directly approximating the linearizing control law \cite{yesildirek1995feedback, sahin_learning_2016, umlauft2017feedback, westenbroek2020feedback} or modeling the dynamic mismatch from an imperfect linearization \cite{helwa2019provably, greeff_exploiting_2021}. In contrast, we consider the inverse problem of discovering linearizing outputs when the dynamics are perfectly known. Separately, Goswami et al. \cite{goswami_data-driven_2023} aims to discover linearizing transforms via purely data-driven invertible neural networks, which limits downstream control analysis. By contrast, our approach is model-based, solving the governing PDE conditions to discover the desired linear representation while enabling formal tracking error bounds. Other linearizing frameworks, such as Koopman operator approaches \cite{brunton_modern_2022, lusch_deep_2018, pan_physics-informed_2020}, focus largely on system dynamics and do not leverage control-oriented structure like the Frobenius condition.

\subsubsection{Physics-Informed Neural Networks in Control}
By casting PDE solving as optimization over sampled points, PINNs \cite{raissi2019physics} have been successfully applied to control tasks involving Hamilton-Jacobi-Bellman \cite{fotiadis2023physics} and Hamilton-Jacobi-reachability equations \cite{bansal2021deepreach}. However, their application to feedback linearization remains unexplored. A primary barrier is that naively applying PINNs to the high-order PDEs that characterize linearizing transforms introduces significant computational bottlenecks and numerical inaccuracies—challenges we overcome in this work through a novel cascadic training algorithm.

\subsection{Statement of Contributions}
Our contributions are threefold. We present:
\begin{itemize}
    \item A novel, cascaded physics-informed learning method that efficiently approximates linearizing transforms for both single-input single-output (SISO) and multiple-input multiple-output (MIMO) systems;
    \item A realizable controller with learning-aware tracking error bounds that depend explicitly on the quality of the learned coordinate transform, directly linking the learning effectiveness to control performance;  
    \item Extensive empirical validation demonstrating our method's capacity to effectively feedback-linearize a broad class of complex dynamical systems, including a planar quadrotor model and systems for which analytical solutions are intractable.
\end{itemize}

\subsection{Notations}
We use $D_\vcx f$ to denote the derivative of $f$ with respect to $\vcx$. When $f$ and $\vcx$ are both vectors, $D_\vcx f$ denotes the Jacobian. We use the notation $\mathcal{L}_f h(\vcx):= D_\vcx h(\vcx) \cdot f(\vcx)$ to denote the Lie derivative of $h$ along the vector field $f$. We use $\mathrm{ad}_fg:=[f, g] = D_\vcx g \cdot f - D_\vcx f \cdot g$ to denote the Lie bracket between vector fields $f$ and $g$. For convenience, we denote the iterative Lie brackets as $\mathrm{ad}_f^{k+1}g:=[f, \mathrm{ad}_f^{k}g]$ with $\mathrm{ad}_f^0 g = g$.

%% file: sections/2-problem-statement.tex
\section{Problem Formulation} \label{sec:problem-formulation}
Consider the nonlinear, MIMO, control-affine system
\begin{equation}\label{eq:control-affine-sys}
    \dot \vcx = f(\vcx) + g(\vcx)\vcu,
\end{equation}
where state $\vcx \in \R^n$, input $\vcu \in \R^m$, and $f: \R^n \to \R^n, g: \R^n \to \R^{n\times m}$ are smooth vector fields. Given \emph{perfect knowledge} of $f$ and $g$ over an open region $\stD \subseteq \R^n$, our objective is to \emph{learn} a diffeomorphism $\Phi: \stD \to \R^n$ and a nonlinear feedback control law $\vcu = \psi(\vcx, \vcv)$ such that the transformed state $\vcz = \Phi(\vcx)$ follows linear, time-invariant dynamics
\begin{equation}\label{eq:linearized-dynamics}
    \dot\vcz = \mtA \vcz + \mtB \vcv,
\end{equation}
where $(\mtA, \mtB)$ is a canonical Brunovsky pair (i.e. a decoupled chain of integrators).
% \footnote{We note that the choice of Brunovsky canonical form is without loss of generality, as any controllable pair $(\mtA, \mtB)$ can be transformed into this form through a linear and time-invariant change of coordinates.} 
Henceforth, we assume that such a diffeomorphism $\Phi(\vcx)$ exists and refer to it as the \textit{linearizing coordinate transform}.
% Here, $(\mtA, \mtB)$ is in Brunovsky canonical form (i.e., a decoupled chain of integrators)\footnote{We note that the choice of Brunovsky normal form is without loss of generality, as any controllable pair $(\mtA, \mtB)$ can be transformed into this form through a linear and time-invariant change of coordinates.}. We refer to $\Phi(\vcx)$ as the \textit{linearizing coordinate transform}; when implemented with the above control law, the transform is said to \textit{exactly} feedback linearize the dynamics on $\stD$. While we do not know the form of $\Phi(\vcx)$, we assume that the systems we consider are feedback linearizable, i.e., that such a coordinate transform $\Phi(\vcx)$ exists.
\par The paper is organized as follows: PDE conditions for feedback linearization are posed in Section~\ref{sec:FBL}, motivating a physics-informed learning framework developed in Section~\ref{sec:learning}. Section~\ref{sec:controller} bounds the effects of learning error on downstream tracking performance; Section~\ref{sec:experiment} presents extensive experimental validation of our approach.

%% file: sections/3-fbl.tex
\section{Background on Feedback Linearization}\label{sec:FBL}
The theory of feedback linearization yields a set of PDEs whose solution directly determines the linearizing coordinate transform $\Phi$. For clarity, we restrict our exposition to the SISO case ($m=1$);  the MIMO generalization is presented in the Appendix~\ref{appendix:mimo} of the full version \cite{fullversion}.
% Since relevant conditions in the MIMO setting closely resemble those in the SISO setting ($m=1$, where $g(\vcx)$ is a single vector field), we defer the full MIMO treatment to Appendix~\ref{appendix:mimo} and restrict our primary exposition to the SISO setting to maintain notational clarity. 
For SISO systems, existence of a linearizing coordinate transform is equivalent to that of a smooth, real-valued \textit{linearizing output} $h(\vcx)$ with relative degree $n$ on the domain $\stD$ (see Lemma 4.2.1 in \citet{isidori1985nonlinear}). Thus, for all $\vcx \in \stD$, $h(\vcx)$ satisfies:
\begin{enumerate}[label=\textbf{CA\arabic*}, ref=CA\arabic*, leftmargin=*]
    \item \label{cond:C1} $\mathcal{L}_g\mathcal{L}_{f}^i h(\vcx)=0$ for $i\in\{0,\ldots,n-2\}$,
    \item \label{cond:C2} $\mathcal{L}_g\mathcal{L}_{f}^{n-1}h(\vcx)\neq 0$.
\end{enumerate}

The linearizing transform is constructed by direct concatenation of $h(\vcx)$ and its Lie derivatives. Concretely: $\Phi(\vcx) = [\Phi_1(\vcx), \dots, \Phi_n(\vcx)]^\top$, where
\begin{equation}\label{eq:diffeo}
    \Phi_i(\vcx) = \lie^{i-1}_f h(\vcx).
    % \Phi(\vcx) = [h(\vcx), \mathcal{L}_f h(\vcx), \dots, \mathcal{L}_f^{n-1} h(\vcx)]^\top.
\end{equation}
The corresponding linearizing control law is given as
\begin{equation}\label{eq:siso-control-law}
    \vcu = \psi(\vcx, \vcv) = A_\Phi(\vcx)^{-1}(\vcv - b_\Phi(\vcx)),
\end{equation}
where $A_\Phi(\vcx) = {\mathcal{L}_g\Phi_n(\vcx)}$ and $b_\Phi(\vcx) =\lie_f\Phi_n(\vcx)$. \textbf{CA} readily defines a set of PDE conditions on $h(\mathbf{x})$ whose solution fully parameterizes the desired transform $\Phi$ via \eqref{eq:diffeo}. However, learning this solution is challenging due to the presence of high-order Lie derivatives. To reduce the order, the Lie derivatives are recast in terms of iterated Lie brackets to get the following set of equivalent, first-order conditions:
\begin{enumerate}[label=\textbf{CB\arabic*}, ref=CB\arabic*, leftmargin=*]
    \item \label{cond:C3}
    $\lie_{\adfg{i}}h(\vcx)=0$ for $i\in\{0,\ldots,n-2\}$,
    \item \label{cond:C4}
    $\lie_{\adfg{n-1}}h(\vcx) = (-1)^{n-1}\lie_g\lie_f^{n-1}h \neq0$.
\end{enumerate}
% Lie brackets are intrinsically linked to the geometric structure of a dynamical system and provide the fundamental analytical framework for assessing the feasibility of exact feedback linearization.
While \eqref{cond:C3} constitutes a set of $n-1$ equalities
\begin{equation}
D_{\vcx}h\cdot\begin{bmatrix}
g& \adfg{} &\hdots&\adfg{n-2}
\end{bmatrix}=0
\end{equation}
that are well-suited to regression (see Section \ref{sec:learning}), the non-triviality condition \eqref{cond:C4} ensuring the numerical stability of the control law \eqref{eq:siso-control-law} is underconstrained. One approach, consistent with~\eqref{cond:C3}, is to enforce 
\begin{equation}\label{eq:equality-reformulation}
    \lie_{\adfg{n-1}}h(\vcx) = c
\end{equation}
for a nonzero constant $c$ over the domain $\stD$. However, this equality formulation is restrictive; a solution can fail to exist for \eqref{eq:equality-reformulation} even when one exists for \textbf{CB}. In Section~\ref{sec:learning} we propose a more general method to reformulate this condition. Concluding the theoretical background, we note that our method also applies naturally to the MIMO setting, deferring further details to Appendix~\ref{appendix:mimo}.

%% file: sections/4-learning.tex
\section{Learning Linearizing Transformations}\label{sec:learning}
We propose a PINN framework to compute the desired diffeomorphism $\Phi$ by approximating solutions to the linearizing PDE conditions. We first present a direct Lie bracket-based algorithm for learning the linearizing output, and show that constructing the diffeomorphism through high-order differentiation of this output introduces severe compounding errors. We then introduce a novel cascaded neural network structure that sequentially minimizes tractable, first-order residuals to mitigate these errors and enable effective downstream control.
% We propose a PINN framework to compute the desired diffeomorphism $\Phi$ by approximating solutions to the PDE conditons \eqref{cond:C1} and \eqref{cond:C2}. Leveraging their first-order reformulation in \eqref{cond:C3} and \eqref{cond:C4}, we first examine a purely bracket-based single-network approach and show that it suffers from compounding errors when approximating higher-order derivatives. To address this limitation, we introduce a novel cascadic structure of neural networks that sequentially minimizes tractable first-order residuals.

\subsection{Learning Linearizing Outputs}
PINNs \citet{raissi2019physics} compute approximate PDE solutions by parameterizing the candidate function as a deep neural network and minimizing the PDE residuals via a standard loss function. Defining the candidate linearizing output $h_\theta(\vcx)$ with network weights $\theta$, the consequent PDE residuals of \eqref{cond:C3} are compactly expressed as
\begin{equation*}
\mathcal{RB}_\theta(\vcx):=D_\vcx h_\theta\cdot
\begin{bmatrix}
g& \adfg{}&\hdots&\adfg{n-2}
\end{bmatrix},
\end{equation*}
while the non-triviality residual of \eqref{cond:C4} is expressed as
\begin{equation*}
    \mathcal{SB}_\theta(\vcx) :=D_\vcx h_\theta\cdot \adfg{n-1}.
\end{equation*}\par
We appeal to the first-order Lie bracket-based conditions \textbf{CB}, to circumvent the computational complexity introduced by the high order Lie derivatives in \textbf{CA}.
Since $\mathcal{D} \subset \mathbb{R}^n$ is compact and $\mathcal{SB}(\vcx)=D_\vcx h\cdot \adfg{n-1}$ is continuous for any continuously differentiable $h(\vcx)$, we can always pick $a,b>0$ as hyperparameters and find a linearizing output s.t. $\mathcal{SB}(\vcx)\in[a,b]$ for all $\vcx\in\mathcal{D}$, ensuring non-triviality while preserving regularity (e.g., to avoid excessively large control gains).
 We can now train the network $h_\theta(x)$ to approximate such an output using the composite loss
 \begin{equation}\label{eq:loss-1}
     L(\theta)=L_\mathcal{RB}(\theta)+L_\mathcal{SB}(\theta),
 \end{equation}
 with $L_\mathcal{RB}(\theta)=\sum_{i=1}^{N} ||\mathcal{RB}_\theta(x_i)||^2_2$ to enforce the desired bracket annihilation and $L_\mathcal{SB}(\theta)=\sum_{i=1}^{N}(\max(a-\mathcal{SB}_\theta(\vcx_i),0)+\max(\mathcal{SB}_\theta(\vcx_i)-b,0))^2$ to impose the box constraints preventing triviality. 
 \begin{comment}
 \footnote{\color{blue} The linearizing output property is scale-invariant. Hence, we can always find a linearizing $h(\vcx)$ that is tightly lower bounded by $a$. Then, $b$ can simply be picked as a (non-tight) upper bound of the aforementioned function $h$ with the sole purpose of preventing excessive control gains. In practice, both $a$ and $b$ can be selected as hyperparameters without further implications.}
 \end{comment}

% We note that these losses are \emph{consistent} with \eqref{cond:C3} and \eqref{cond:C4} in the sense that any true linearizing output of the system incurs zero loss. However, we observe that a neural network trained under the supervision of \eqref{eq:loss-1} can only approximate a true linearizing output well up to its first derivative; hence, its actual applicability in downstream control is undermined.

\iffalse
We note that these losses are \emph{consistent} for function approximation, in the sense that any true linearizing output of the system will incur zero loss under~\eqref{eq:loss-1}, and thus belongs to the set of minimizers. However, because~\eqref{eq:loss-1} is first-order in $h_{\theta}(\vcx)$, it is ill-suited in practice for the downstream control task, which we elucidate further in the following section. 
\fi
 
\subsection{Difficulty in Constructing $\Phi$ via Differentiation}\label{sec:limits-brackets}
While the loss function \eqref{eq:loss-1} serves as an effective supervision for finding linearizing outputs, constructing the diffeomorphism $\Phi$ requires the complete set of Lie derivatives that appear in \textbf{CA}. However, we observe that a neural network trained under loss \eqref{eq:loss-1} can only approximate a true linearizing output well up to its first derivative~\textendash~the highest order derivative that is explicitly supervised. This limitation is closely related to the deep learning phenomenon of \emph{spectral bias}~\citet{rahaman2019spectralbiasneuralnetworks}, which describes the tendency of neural networks to preferentially approximate low-frequency components of a target function. This bias is detrimental to the approximation of higher order Lie derivatives, since differentiation amplifies errors proportionally to frequency. 

% To effectively employ the approximate linearizing output $h_\theta(\vcx)$ in actual control tasks, one requires access to the complete set of Lie derivatives appearing in~\eqref{eq:diffeo}, which must satisfy \eqref{cond:C1} and \eqref{cond:C2}. Unfortunately, this cannot be ensured by imposing point-wise soft constraints solely on first-order derivatives of the neural network $h_\theta(\vcx)$ via~\eqref{eq:loss-1}. Intuitively, this limitation is closely related to the deep learning phenomenon of \emph{spectral bias}~\citet{rahaman2019spectralbiasneuralnetworks}, which describes the tendency of neural networks to preferentially approximate low-frequency components of a target function. This bias is detrimental to the approximation of higher-order derivatives, since differentiation amplifies errors proportionally to frequency. 

%As such, naively applying standard learning algorithms will result in good function approximation (which uses low-frequency metrics), but poor control performance (which depends on the high-frequency content of $h(\vcx)$).

% In addition to this general interpretation, the degradation in the accuracy of higher-order Lie derivatives can also be attributed to the following lemma.
Additionally, the small first-order bracket errors \eqref{eq:loss-1} in the learned outputs can be amplified via high-order differentiation, as illustrated in the following lemma.
% which separately implies that low-frequency errors will %tend to compound. 
\begin{lemma}
    Let $e_i(\vcx):=\lie_{\adfg{i}}h(\vcx)$, for $\vcx\in \mathcal{D}$, be the residuals of \eqref{cond:C1}. We also define the residuals in \eqref{cond:C3} as $\epsilon_i(\vcx):=\lgfh{i}(\vcx)$. Assuming known $e_i$ for $i\in\{0,\hdots,k\}$, $\epsilon_i$ is computed with the iterative formula:
    \begin{equation*}
        \epsilon_i=(-1)^{i}e_i+\sum_{r=0}^{i-1}(-1)^{r+i+1}\binom{i}{r}\lie_f\epsilon_{r}.
    \end{equation*}
\end{lemma}
A direct consequence of Lemma 4.1.2 in \citet{isidori1985nonlinear}, this result formalizes the attenuation of implicit supervision over~\eqref{cond:C1} residuals $e_i$ when training on~\eqref{cond:C3}, as the minimization targets $\epsilon_i$ come to be dominated by the accumulating right-most terms. Thus, the direct approach only weakly penalizes violations of higher-order Lie derivative conditions, compromising the reliability of the learned transform in downstream control tasks.

% highlights how residuals arising from first-order bracket-based constraints propagate and accumulate through successive Lie derivatives, resulting in amplified violations of higher order conditions and ultimately compromising the reliability of the learned transform in downstream control tasks. 

To mitigate this, one could augment the loss in \eqref{eq:loss-1} with explicit supervision of higher order Lie derivatives via the conditions \textbf{CA}. However, incorporating high-order Lie derivatives introduces significant computational and optimization challenges. Although the former can be addressed to some extent by Taylor-mode automatic differentiation \citet{bettencourt2019taylor}, the latter make this approach impractical as system complexity grows (see Section~\ref{sec:experiment}).

\subsection{Cascadic Learning of the Diffeomorphism}\label{sec:cascade}

To overcome the limitations of direct bracket-based supervision, we adopt a cascaded representation using separate neural networks
\begin{equation}
\hat{\Phi}(\vcx):=\begin{pmatrix}h_{\theta_1}(\vcx),\hdots,h_{\theta_n}(\vcx)\end{pmatrix},
\end{equation}
parametrizing the linearizing coordinate transform $\hat{\Phi}(\vcx)=\begin{pmatrix}\hat{\Phi}_1(\vcx), \hdots, \hat{\Phi}_n(\vcx)\end{pmatrix}$. Exploiting the ability of the bracket-based approach to approximate a linearizing output up to its first order, we combine this cascaded structure with a curriculum training strategy. In particular, we first train $h_{\theta_1}(\vcx)$ using \eqref{eq:loss-1}. Then, for $i=2,\dots,n$, we train each $h_{\theta_i}(\vcx)$ in succession by minimizing
\begin{equation}\label{eq:loss-3}
L(\theta_i)=L_{\mathcal{F}}(\theta_i)+L_{\mathcal{G}}(\theta_i),
\end{equation}
where the consistency term
\begin{equation}
L_{\mathcal{F}}(\theta_i)
=
\sum_{j=1}^N
\bigl(
\lie_f h_{\theta_{i-1}}(\vcx_j) - h_{\theta_i}(\vcx_j)
\bigr)^2
\end{equation}
encourages reconstruction of the Lie derivative of the preceding layer $h_{\theta_{i-1}}$, while the control-specific term
\begin{equation}
L_{\mathcal{G}}(\theta_i)
=
\sum_{j=1}^N
\bigl(\lie_g h_{\theta_i}(\vcx_j)\bigr)^2
\end{equation}
promotes adherence to conditions \textbf{CA}.

Intuitively, the cascadic approach decomposes the original learning problem---where a single network must satisfy conditions involving Lie derivatives up to order $n$---into a sequence of simpler tasks involving only first-order derivatives across multiple networks. By relaxing the implicit constraint that $\hat{\Phi}_{i+1}(\vcx)=\lie_f \hat{\Phi}_i(\vcx)$, the resulting coordinate transform satisfies \textbf{CA} more tightly than approaches relying solely on bracket-based supervision. Moreover, this formulation avoids the optimization pathologies associated with explicitly supervising high-order Lie derivatives. Empirically, the cascadic approach demonstrates superior performance and reliability compared to single-network formulations.%, as evidenced by experiments on multiple complex systems. 

\begin{algorithm}[t]
\caption{Cascadic Training of the Linearizing Transform}
\label{alg:cascadic-training}
\begin{algorithmic}[1]

\STATE \textbf{Input:} Training set $\{\vcx_j\}_{j=1}^{N}$

\STATE \textbf{Output:} Learned transform
$\hat{\Phi}(\vcx)
=
\bigl(
h_{\theta_1}(\vcx),\ldots,h_{\theta_n}(\vcx)
\bigr)$

\STATE Initialize $\theta_1,\ldots,\theta_n$

\STATE Train $h_{\theta_1}$ according to~\eqref{eq:loss-1}
\STATE Freeze $h_{\theta_1}$

\FOR{$i=2,\ldots,n$}
    \STATE Train $h_{\theta_i}$ according to~\eqref{eq:loss-3}
    \STATE Freeze $h_{\theta_i}$
\ENDFOR

\STATE \textbf{return}
$\hat{\Phi}(\vcx)
=
\bigl(
h_{\theta_1}(\vcx),\ldots,h_{\theta_n}(\vcx)
\bigr)$

\end{algorithmic}
\end{algorithm}

We conclude with a brief remark on computational complexity. The cascadic approach requires only first-order derivatives for training, significantly reducing the computational cost of constructing high-order derivatives as compared to direct approaches while yielding linear scaling with the number of cascade stages. Further, in MIMO settings only the decoupling matrix requires joint training, and the remaining layers can be trained independently to admit straightforward parallelization.

%% file: sections/5-controller.tex
\section{Tracking Control with Learned Coordinate Transform}\label{sec:controller}
To illustrate how learning errors manifest in downstream control performance, we consider the task of following a dynamically feasible and bounded reference trajectory $(\bar{\vcx}, \bar{\vcu})$ under the nonlinear dynamics \eqref{eq:control-affine-sys}. With the learned, approximately linearizing, transform $\hat{\Phi}$, we design an error-feedback controller in the linearized space. We characterize the tracking error of this controller in terms of learning losses. Before proceeding, we make the following assumptions on the map $\hat{\Phi}$.\footnote{Empirically, we observe no invertibility issues in the examples considered. We further analyze in the Appendix how training errors affect the invertibility of the learned transform relative to the assumed ground-truth linearizing diffeomorphism.}
\begin{assumption}\label{assm:lipschitz-T}
    $\hat{\Phi}$ is invertible. Further, there exists a positive constant $l_{\hat{\Phi}}$ such that $l_{\hat\Phi} \norm{\vcx_1 - \vcx_2} \leq \norm{\hat{\Phi}(\vcx_1) - \hat{\Phi}(\vcx_2)}$ for all $\vcx \in \stD$.
\end{assumption}
The invertibility of $\hat\Phi$ is necessary to relate the tracking error in the linearized domain to that in the nonlinear domain and is a standard assumption for analyzing approximate linearizing transforms \cite{johansen2000computational}. The existence of $l_{\hat\Phi}$ ensures that the inverse is not arbitrarily ill-conditioned.

Let $\bar{\vcz} = \hat{\Phi}(\bar{\vcx})$ be the reference linearized state and $\bar \vcv:= \stL_f \hat{\Phi}_{n-1}(\bar\vcx) + \stL_g \hat{\Phi}_{n-1}(\bar\vcx)\bar \vcu$ as the reference virtual input. We define the control law in the linearized space:
\begin{equation*}
    \vcv(t) = \bar{\vcv}(t) - \mtK(\vcz - \bar{\vcz}(t)),
\end{equation*}
where $\mtK$ is a robust stabilizing gain matrix
% for the Brunovsky pair $(\mtA, \mtB)$
designed by solving the $\stH_\infty$ Riccati equation:
\begin{gather}\label{eq:h-infty-riccati}
    \mtP\mtA + \mtA^\top \mtP - \mtP\mtB\mtR^{-1}\mtB^\top \mtP = -\mtQ - \frac{1}{\gamma^2}\mtP\mtP.
\end{gather}
Here, $(\mtA, \mtB)$ is the Brunovsky pair, and $\mtQ$,  $\mtR$ are properly chosen cost matrices. The gain matrix is then taken to be $\mtK = \mtR^{-1} \mtB \mtP$. Applying the linearizing control law \eqref{eq:siso-control-law}, the control input is given as
\begin{equation}\label{eq:control-law}
    \vcu = A_{\hat\Phi}^{-1}(\vcx)\left(\bar v - \mtK\left(\vcz - \bar \vcz\right) - b_{\hat\Phi}(\vcx)\right).
    % -\frac{\stL_f\hat{\Phi}_{n-1}(\vcx) - (\bar v(t) + K(\hat{\Phi}(\vcx) - \bar \vcz(t)))}{\stL_g\hat{\Phi}_{n-1}(\vcx)}.
\end{equation}
Because of the existence of learning errors, $\hat{\Phi}$ does not perfectly linearize the $\vcz$-space dynamics and leads to the \textit{perturbed} $\vcz$-dynamics
\begin{equation}\label{eq:perturbed-z}
\dot{\vcz} = \mtA \vcz + (\mtB + \Delta_B(\vcx)) \vcv + \vcd_z(\vcx),
\end{equation}
where the multiplicative input disturbance $\Delta_B(\vcx)$ and the drift $\vcd_z(\vcx)$ arise from $\hat{\Phi}$ not perfectly satisfying conditions \textbf{CA}. Two sources of learning inaccuracy drive these terms: the PDE residuals $r_i(\vcx) := \mathcal{L}_g \hat{\Phi}_i(\vcx)$, and the cascade consistency errors $\vcd_c(\vcx)$, with $i$-th element $\vcd_{c,i}(\vcx) := \mathcal{L}_f \hat{\Phi}_{i-1}(\vcx) - \hat{\Phi}_i(\vcx)$. Denoting the stacked residual ${\vcr(\vcx) := [r_1(\vcx), \cdots, r_{n-2}(\vcx), 0]}$, the perturbation terms can be expressed as
\begin{equation}\label{eq:control-residual-form}
\begin{gathered}
    \Delta_B(\vcx):= \vcr(\vcx)A_{\hat\Phi}^{-1}(\vcx),\\
    \vcd_z(\vcx):= \vcr(\vcx)A_{\hat\Phi}^{-1}(\vcx)b_{\hat\Phi}(\vcx) + \vcd_c(\vcx).
\end{gathered}
\end{equation}
Note that the PDE residuals $\vcr$ and consistency error $\vcd_c$ are explicitly minimized during training via corresponding loss terms in \eqref{eq:loss-3}. Further, given that $A_{\hat\Phi}^{-1}(\vcx)$ is invertible (promoted via the term $L_\stS$ in \eqref{eq:loss-1}), $\norm{\Delta_B(\vcx)}_F$ and $\norm{\vcd_z(\vcx)}_2$ vanish as $\vcr$ and $\vcd_c$ approach zero. Thus, when the learning error is small, it is reasonable to assume that the disturbance terms \eqref{eq:control-residual-form} are small and bounded over the training domain. 
\begin{assumption}\label{assm:bounded-residual}
    We assume that for all $\vcx \in \stD$, the disturbance terms $\norm{\Delta_B(\vcx)}_F \leq \bar\delta,\; \norm{\vcd_z(\vcx)}_2 \leq \bar d.$
\end{assumption}
Finally, to ensure safety, the reference trajectory must maintain a sufficient margin strictly within the domain where our learning guarantees hold.
\begin{assumption}\label{assm:robust-reference-main-text}
The reference virtual input is bounded by $\|\bar{\vcv}(t)\| \leq \bar{V}$, and the reference states $\vcx(t)$ remain strictly inside the interior of the training domain $ \mathcal{D}$.
\end{assumption}

In the interest of space, we present an informal version of the result, deferring the formal statement and proof to Appendix~\ref{appendix:proof} of the full version \cite{fullversion}.
\begin{theorem}[informal]\label{thm:tracking}
    Suppose Assumptions \ref{assm:lipschitz-T}--\ref{assm:robust-reference-main-text} hold and that the input perturbation is sufficiently small such that $\bar{\delta} < \sqrt{\lmin(\mtR)}/\gamma$. Then, there exists a forward-invariant set $\Omega_x=\{\vce\mid\norm{\vce}_2^2\leq r_{\vce}^2\}$ for the state tracking error $\vce = \vcx - \bar{\vcx}$. The radius $r_{\vce}$ is bounded by the additive and multiplicative learning errors and scales as:
    \begin{equation}
        r_\vce^2 = \mathcal{O}\left( \frac{{\bar{d}}^{2} + \bar{\delta}^2 \bar{V}^2}{\lmin(\mtR) - \gamma^2 {\bar \delta}^2} \right).
    \end{equation}
    Consequently, if the initial state satisfies $\vce_x(0) \in \Omega_x$, and the reference trajectory maintains a safety margin from the boundary of $\mathcal{D}$ such that $\bar{\vcx}(t) \in \mathcal{D} \ominus \Omega_x$, the true state $\vcx(t)$ never leaves $\mathcal{D}$, guaranteeing that the local error bounds hold for all $t \geq 0$.
\end{theorem}
Theorem \ref{thm:tracking} formalizes the relationship between the learning errors and closed-loop control. Driving the PDE residuals $\vcr(x)$ and consistency errors $\vcd_c(\vcx)$ to zero during training directly shrinks $\bar{\delta}$ and $\bar{d}$, thereby tightening the guarantees on downstream tracking performance. In Section~\ref{sec:experiment}, we approximate these worst-case bounds by densely sampling the training domain and taking the maximum observed values of $\delta(\vcx)$ and $d(\vcx)$. Extending PINN generalization analyses such as~\cite{mishra2023estimates} to obtain certified bounds between samples is an important direction of future work.

% By leveraging the robustness of the linear controller, we can show that when the disturbance terms $\Delta_B$ and $\vcd$ are sufficiently small, the tracking error of the resulting controller can be guaranteed when the reference trajectory is sufficiently within the domain $\stD$.
% \begin{assumption}\label{assm:robust-reference-main-text}
%     The reference trajectory $(\bar\vcx, \bar u)$ satisfies that the reference virtual input is bounded
%     $$\norm{\bar v} \leq \bar V,$$
%     and that the reference trajectory lies within a robustified version of the domain $\stD$,
%     $$\vcx(t) \in \stD \ominus \Omega_x,$$
%     where $\ominus$ denotes Minkowski difference and $\Omega_x$ is an error set given as
%     \begin{equation*}
%         \Omega_x := \left\{ \vce_x \mid \norm{\vce_x}^2 \leq \frac{4\kappa(P)R (\bar d + \bar\delta \bar V)^2}{l_T^2 \lmin(Q)\left(1 + \frac{2R}{\gamma^2} - \sqrt{1 + 4\bar\delta^2}\right)} \right\}.
%     \end{equation*}
%     Here $\kappa(\cdot)$ denotes the condition number of the given matrix.
% \end{assumption}

%% file: sections/6-experiments.tex
\section{Experiments}\label{sec:experiment}
We test our learning algorithm and the corresponding learned downstream controllers in an extensive array of feedback linearizable systems. We first test our method on well-known examples that admit analytic linearizing output solutions; these instances have simple (e.g., affine) forms that are easy to learn. Then, to further stress-test our approach, we synthesize complex SISO and MIMO systems for which analytical solutions are intractable. The results show that our approach can achieve good empirical performance for systems of various sizes in terms of both the quality of the linearizing transform and the controller performance.

% note that since the linearizing outputs of real systems are often derived by hand (Section~\ref{sec:real-sys} and also \citet{goswami_data-driven_2023, westenbroek2020feedback,umlauft2017feedback}), they tend to have simple linearizing outputs and constitute easy learning tasks. Thus, to test the efficacy of the method in finding such coordinate transforms, where hand derivations can fail because of the systems' complexity, we also consider both SISO and MIMO synthetic systems. The results show that our approach can achieve good empirical performance for systems of various sizes in terms of both the quality of the linearizing transform and the controller performance.
%Before proceeding to the numerical evaluation of our proposed algorithm, we need to introduce the method deployed to generate synthetic feedback linearizable systems and a metric designed to uniformly evaluate alternative learning strategies.

\subsection{Learning Setup and Evaluation}
We employ Jaxpi \citet{wang2023expertsguidetrainingphysicsinformed} as our PINN implementation framework and adopt the modified MLP architecture proposed in~\citet{wang2020understandingmitigatinggradientpathologies} for our candidate PDE solutions, as it exhibits superior performance compared to standard feed-forward MLPs. 
As a baseline, we implement the \textbf{direct} single-network approach that learns a linearizing output supervised by both conditions \textbf{CA} and \textbf{CB}.

\begin{comment}
Specifically, we parameterize the modified MLPs with three hidden layers of 256 neurons and GELU activation functions. The network weights are initialized using the scheme suggested in~\citet{wang2022randomweightfactorizationimproves} and training is performed using the ADAM optimizer for $200$--$300\times 10^3$ iterations for the first and $150$--$200\times 10^3$ for the subsequent layers of the cascade.  % \red{We refer to this method as ??? in what follows.}
%$180$--$240\times 10^3$ iterations for the first layer of the cascade, and $90$--$180\times 10^3$ iterations for subsequent layers.
\end{comment}

In addition to reporting the metrics explicitly included in the losses supervising training, we seek to directly measure how well the learned coordinate transform linearizes the system dynamics. To do so, we sample $\vcx$ uniformly in $\stD$, $\vcu \sim \stN(\bf{0}, \mtI)$ and compute the difference between the true time derivative of the transform
$\dot{\hat\Phi}(\vcx) = D_\vcx \hat\Phi(\vcx) (f(\vcx) + g(\vcx) \vcu)$
and the time derivative predicted by the linear dynamics
$\dot{\hat\Phi}^{pred}(\vcx) = \mtA \hat\Phi(\vcx) + \mtB (\lie_f \hat\Phi_{n-1}(\vcx) + \lie_g \hat\Phi_{n-1}(\vcx) \vcu)$.
We then measure the quality of linearization as the relative error of the true and approximate time derivatives:
\begin{equation} \tag{RLQE}
    % \varepsilon(\vcx) := \frac{\norm{\dot{\hat\Phi}(\vcx) - \dot{\hat\Phi}^{pred}(\vcx)}}{\norm{\dot{\hat\Phi}(\vcx)}}.
    \varepsilon(\vcx) := \norm{\dot{\hat\Phi}(\vcx) - \dot{\hat\Phi}^{pred}(\vcx)}{\big /}\norm{\dot{\hat\Phi}(\vcx)}.
\end{equation}
This \emph{relative linearization quality error} (RLQE) is especially useful for comparing our proposed cascadic approach and the single network baseline, as the former intentionally introduces a cascade inconsistency between $T_{i+1}(\vcx)$ and $\lie_fT_i(\vcx)$ but the latter does not, making their PDE residual not directly comparable. In contrast, RLQE provides a fair comparison as it directly quantifies the quality of the downstream linearization. Finally, we report the \emph{folded fraction}: the proportion of sampled states at which the sign of $det(D_\vcx\Phi)$ differs from the dominant sign of the rest of the samples, implying that the learned transform ceases to be a diffeomorphism.

\subsection{Controller Setup and Evaluation}\label{sec:control-setup}
Once the linearizing coordinate transform is learned, we construct the downstream tracking controller following the procedure in Section~\ref{sec:controller} with $\mtQ:=q\mtI$ and $\mtR:=\mtI$, where $q \in [1, 5]$ is system-specific. To evaluate the controller, we generate 64 reference trajectories for each system by using a ground-truth feedback linearization controller to regulate the state to the origin from initial conditions randomly sampled from the interior of the domain. Example reference trajectories are visualized in Figures~\ref{fig:trajectories} and~\ref{fig:xz-comparison}. To test robustness, we apply the tracking controller from perturbed initial conditions, with perturbations sampled from an $\ell_\infty$ ball of radius $0.15$. We simulate the controllers for $T \in [8, 10]$ seconds with RK45 numerical integration and a step size of $0.002$s.

% via a nominal linearizing output--considered unknown during training--and an LQR-based controller that tries to stabilize the system in the linearized state space. Typically, for this controller, we consider a small state cost, compared to the one in the corresponding tracking controller, to prevent the generated trajectories from  directly converging to the origin. Furthermore, we ensure that the reference trajectories stay within the bounds of the training domain and we randomly perturb their initial states to acquire the initial conditions for the controlled system.

We evaluate the tracking controller with the \emph{average tracking error} ($\frac{1}{T} \int_0^T\norm{\vcx - \bar\vcx} \mathrm{d}t$), the \emph{final state error} ($\norm{\vcx(T) - \bar\vcx(T)}$), and the \emph{total control effort} ($\int_0^T\norm{\vcu} \mathrm{d}t$). For each system, we report the median errors among the 64 reference trajectories, along with the number of diverging trajectories, to capture potential stability issues. To contextualize the results, we also report the performance of a feedback linearization controller constructed with the ground-truth diffeomorphism.

\subsection{Results: Real Systems} \label{sec:real-sys}
We first test our pipeline on two textbook examples with plain linearizing outputs: a single-link manipulator (Example 6.10 in \citet{slotine1991applied}) and a planar quadrotor.

\subsubsection{Flexible Joint}
The single link manipulator is a fourth order SISO system with a linearizing output $h(\vcx)=x_1$ (see Example 6.10 in \citet{slotine1991applied}). We try the single-network approach of Section~\ref{sec:limits-brackets} and observe that it is sufficient for the minimization of the composite loss, consisting of both bracket and high-order Lie derivative terms. The residuals of conditions \textbf{CA} for the learned linearization output are summarized in  Table~\ref{table:flex1} and we note that they indeed result in negligible RLQE. 
\begin{comment}
\begin{align*}
\dot{x}_1 &= x_2 \; ; \quad 
\dot{x}_2 = \frac{1}{M}\big(-mgl\sin(x_1) - k(x_1 - x_3) \big) \\
\dot{x}_3 &= x_4 \; ; \quad 
\dot{x}_4 = \frac{k}{J}(x_1 - x_3) + \frac{1}{J}u
\end{align*}
\end{comment}

\begin{table}[h]
\scriptsize
\centering
\resizebox{\linewidth}{!}{
\begin{tabular}{ccccc}
\toprule
$|L_g h|$ & |$L_g L_f h|$ & |$L_g L_f^2 h$| & $L_g L_f^3 h$ & RLQE \\
\midrule
$4.3 \times 10^{-7}$ & $1.6 \times 10^{-7}$ & $9.1 \times 10^{-7}$ & $1.101$ & $\mathbf{1.4\times10 ^{-8}}$\\
\bottomrule
\end{tabular}}
\caption{Mean Lie derivatives and RLQE for the learned output for the direct method in the flexible joint example.}
\label{table:flex1}
\end{table}
\begin{comment}
\begin{table}[h]
\scriptsize
\centering
\resizebox{\linewidth}{!}{
\begin{tabular}{ccccc}
\toprule
$|L_g h|$ & |$L_g L_f h|$ & |$L_g L_f^2 h$| & $L_g L_f^3 h$ & RLQE \\
\midrule
$2.2 \times 10^{-6}$ & $4.3 \times 10^{-6}$ & $2.1 \times 10^{-5}$ & $1.006$ & $\mathbf{5.4\times10 ^{-7}}$\\
\bottomrule
\end{tabular}}
\caption{Mean Lie derivatives and RLQE for the learned output for the cascade method in the flexible joint example.}
\label{table:flex3}
\end{table}
\end{comment}
For the tracking task, we report the performance of both learned and ground truth transforms in the domain $\mathcal{D}=[-\pi,\pi]\times[-5,5]\times[-\pi,\pi]\times[-5,5]$ in Table~\ref{table:flex2}. The learned linearizing outputs demonstrate robustness to small numerical training errors, captured by the RLQE, achieving nearly identical performance metrics as the ground truth controller.
\begin{table}[h]
\centering
\begin{tabular}{lccc}
\toprule
\textbf{Metric} & \textbf{Ground Truth} &\textbf{Direct} &\textbf{Cascadic} \\
\midrule
RLQE & --- & $1.4\times10^{-8}$ & $5.4\times10^{-7}$\\
Folded fraction & --- & $0.00$\% & $0.00$\%\\
\midrule
Avg. error & 1.000 & 1.003 &  1.010\\
%Avg. $z$-error & 0.260 & 0.303 \\
Final error & 0.004 & 0.004& 0.004\\
Control effort & 605.234 & 605.333 & 603.026 \\
\bottomrule
\end{tabular}
\caption{Tracking performance of the nominal and learned transforms in the single link manipulator system.}
\label{table:flex2}
\end{table}
\subsubsection{Planar Quadrotor}
Another system known to be differentially flat and, therefore, dynamically feedback linearizable is the planar quadrotor in extended dynamics form, represented by the state equations
\begin{align*}
\dot{p}_1 &= v_1; & \dot{v}_1 &= -\frac{F}{m}\sin(\theta); & \dot{F} &= \nu; & \dot{\nu} &= u_1 \\
\dot{p}_2 &= v_2; & \dot{v}_2 &= \frac{F}{m}\cos(\theta) - g; & \dot{\theta} &= \omega; & \dot{\omega} &= \frac{1}{I}u_2
\end{align*} in the square geometric domain $\mathcal{D} = [-5,5]^2$. This system admits the planar coordinates $p_1$, $p_2$ as linearizing outputs, each with relative degree four.
For brevity, we omit the exposition of the 16 residuals related to conditions \textbf{CC} and instead focus on the performance of the learned transforms in the downstream tracking task. These results are summarized in Table~\ref{tab:planar-quad}. 
%Due to space limitations, we omit the exposition of linearization quality metrics regarding MIMO systems and focus on the performance of the learned transforms in the downstream tracking task.

\begin{table}[h]
\centering
\begin{tabular}{lccc}
\toprule
\textbf{Metric} & \textbf{Ground Truth} &\textbf{Direct}  &\textbf{Cascadic} \\
\toprule
Folded fraction & --- & $0.000$\% & $0.000$\%\\
\midrule
Avg. error & 0.186 & 0.186 & 0.186 \\
%Avg. $z$-error & 0.303 & 0.136 \\
Final error & 0.002 & 0.002 & 0.002\\
Control effort & 149.841 & 149.860 & 150.488\\
\bottomrule
\end{tabular}
\caption{Tracking performance of the nominal and learned transforms in the planar quadrotor system.}
\label{tab:planar-quad}
\end{table}

\begin{comment}
We observe that, despite the increased dimension, \rewrite{a}{the naive} single-network-per-output approach is sufficient \rewrite{for the discovery of an output which can}{to near-}perfectly emulate the \rewrite{analytic nominal}{analytic} linearizing output.However, \rewrite{we would like to further examine the empirical performance of both learned models in settings where linearization cannot be achieved to almost machine precision.}{a primary motivation of this work is to learn unknown and analytically opaque linearizing outputs, requiring evaluation on more complex systems.}
\end{comment}
\begin{comment}
\begin{figure}[h]
\centering
\begin{subfigure}{0.2\textwidth}
    \centering
    \includegraphics[width=\linewidth, keepaspectratio]{assets/planar_quad1.png}
    \caption{}
\end{subfigure}
\hfill
\begin{subfigure}{0.2\textwidth}
    \centering
    \includegraphics[width=\linewidth, keepaspectratio,]{assets/planar_quad2.png}
    \caption{}
\end{subfigure}
\caption{Trajectory tracking performance of learned and ground truth models in the planar quadrotor system. Markers denote the initial states of trajectories.}
\label{fig:trajectories}
\end{figure}
\end{comment}

\subsection{Results: Synthetic Systems}\label{subsec:synthetic}
To test our proposed algorithm on more challenging systems, we consider a parameterized class of synthetic systems. In the SISO case, such systems can be generated by applying a diffeomorphism $\vcz = \tilde\Phi(\vcx)$ and an affine input transformation $\vcu = \alpha(\vcx) + \beta(\vcx) \vcv$ to a linear system in Brunovsky canonical form \eqref{eq:linearized-dynamics}. The resulting dynamics functions are given as:
$
f(\vcx) = [D_\vcx \tilde\Phi(\vcx)]^{-1}\big(\mtA \tilde\Phi(\vcx) + \mtB \alpha(\vcx)\big), \;
g(\vcx) = [D_\vcx \tilde\Phi(\vcx)]^{-1} \mtB \beta(\vcx).
$
This procedure ensures linearizability by construction, and the learning problem amounts to reconstructing $\tilde{\Phi}(\vcx)$. MIMO systems can be similarly synthesized with a generalized, vector-valued affine control law. To sample a diverse set of such systems, we parameterize the diffeomorphism $\tilde\Phi$ as an invertible neural network \cite{ardizzone2018analyzing} with randomly initialized weights. More details are provided in Appendix~\ref{app:synthetic} of the full version \cite{fullversion}. We visualize a representative 2D synthetic system in Figure~\ref{fig:xz-comparison} to illustrate the nonlinearity of the synthesized system.
We consider three system configurations: 2D SISO, 4D SISO, and 8D MIMO systems. The 2D SISO setting allows us to validate the evaluation methodology, providing visual insight into the synthesized systems and the overall tracking performance. The 4D SISO systems demonstrate the advantages of the cascadic approach in the presence of higher-order derivatives. Finally, the 8D MIMO systems illustrate the scalability of the proposed method to higher dimensions. For each setting, we generate 10 synthetic systems. We evaluate our learned coordinate transforms on the hyper-cubic domain $\mathcal{D}:=[-1.75, 1.75]^n$, reporting the metrics RLQE and folded fraction (each estimated over 100,000 uniformly sampled points) and the controller metrics introduced in Section~\ref{sec:control-setup}. The aggregate results are summarized in Table~\ref{table:combined}.
% Regarding tracking, we adhere to the methodology proposed in Section~\ref{sec:control-setup}, reporting the mean value of the suggested metrics across all systems for each setting. We also report the fraction of systems with "unstable" behavior as Diverged(any) and the fraction of divergent trajectories across these systems as Perc. diverged. The aggregate results are summarized in Table~\ref{table:combined}.
\begin{table*}[t]
\centering\scriptsize
\begin{tabular}{l|ccc|ccc|ccc}
\toprule
& \multicolumn{3}{c|}{\textbf{2D SISO}} 
& \multicolumn{3}{c|}{\textbf{4D SISO}} 
& \multicolumn{3}{c}{\textbf{8D MIMO}} \\
\textbf{Metric} 
& \textbf{Ground Truth} & \textbf{Direct} & \textbf{Cascadic} 
& \textbf{Ground Truth} & \textbf{Direct} & \textbf{Cascadic} 
& \textbf{Ground Truth} & \textbf{Direct} & \textbf{Cascadic} \\
\midrule
RLQE ($\downarrow$)
& --- & $0.008\%$ & $0.009\%$
& --- & $2.92\%$ & $0.035\%$
& --- & --- & --- \\
Folded fraction ($\downarrow$)
& --- & $0.000\%$ & $0.000\%$
& --- & $11.376\%$ & $0.004\%$
& --- & $1.162$ \% & $0.000$\% \\
\midrule
Avg. Error ($\downarrow$)
& 0.0483 & 0.0483 & 0.0483 
& 0.1172 & 0.1340 & 0.1178& 0.1079 & 0.1210 & 0.1076 \\
Final Error ($\downarrow$)
& 0.0001 & 0.0001 & 0.0001 
& 0.0015 & 0.0017 & 0.0020 
& 0.0008 & 0.0031 & 0.0008 \\
Control Effort ($\downarrow$)
& 45.5854 & 45.5935 & 45.5946 
& 103.6850 & 110.6017 & 103.4540 
& 183.13 & 183.95 & 183.5 \\
Perc. Diverged ($\downarrow$)
& --- & 0.00\% & 0.00\% 
& --- & 4.53\% & 0.00\% 
& --- &0.00\% & 0.00\% \\
\bottomrule
\end{tabular}
\caption{Combined tracking performance across 2D SISO, 4D SISO, and 8D MIMO synthetic systems.}
\label{table:combined}
\end{table*}

\subsubsection{2D SISO Systems}
For 2D SISO systems, both learned coordinate transforms (cascadic and direct) achieve highly RLQE-accurate linearization due to the lack of high-order derivatives. Consequently, the performance of their corresponding controllers closely matches that of the ground-truth controller. Figure~\ref{fig:xz-comparison} illustrates the learned diffeomorphism for a representative system: the box domain $\stD$ in the $\vcx$-space warps nonlinearly in the $\vcz$-space, while the elliptical Lyapunov function native to the linearized $\vcz$-space exhibits nonlinear behavior in the $\vcx$-domain. To contextualize the theoretical results, Figure~\ref{fig:trajectories} visualizes the estimated forward-invariant error set in Theorem~\ref{thm:tracking} with estimated learning error bounds. The minimal learning errors in this case lead to very tight error sets. Further numerical results and visualizations for 2D SISO systems can be found in Appendix~\ref{app:synthetic}.

% When the system is 2D, the two implementations--cascadic and single-network--can be used almost interchangeably because the compounding errors from computing high-order Lie derivatives are mild.
% since the only condition involving a higher order derivative is $\lie_g\lie_fh(\vcx)\neq0$, which is not strict in contrast to equalities encountered in higher order systems. We observe that the learned transforms perform identically to the nominal ones in the downstream tracking task, despite the anticipated increase in the average RLQE compared to the simple real systems and the absence of a custom treatment during the training phase. In addition Table~\ref{table:combined}, Figure~\ref{fig:trajectories} and~\ref{fig:xz-comparison}, further numerical results and visual evidence for 2D SISO systems can be found in Appendix~\ref{app:synthetic}.

\begin{comment}
\begin{table}[h]
\centering
\begin{tabular}{lcc}
\toprule
\textbf{Metric} & \textbf{Ground Truth} &\textbf{ Learned}\\
\midrule
RLQE& --- & $\mathbf{8.52 \times 10^{-5}}$ \\
\midrule
Avg. error & 0.0483 & 0.0483 \\
%Avg. $z$-error & 0.0464 & 0.0472 \\
Final error & 0.0001 & 0.00001 \\
%Final $z$-error & 0.0001 & 0.0001 \\
Control effort & 45.58 & 45.59 \\
\midrule
Diverged (any) & --- & \textbf{0.000} \\
Perc. diverged & --- & \textbf{0.000} \\
\bottomrule
\end{tabular}
\caption{Tracking performance in 2D SISO systems.}
\label{table:2d}
\end{table}
\end{comment}
\begin{figure}[h]
\centering
\begin{subfigure}{0.23\textwidth}
    \centering
    \includegraphics[width=\linewidth, keepaspectratio]{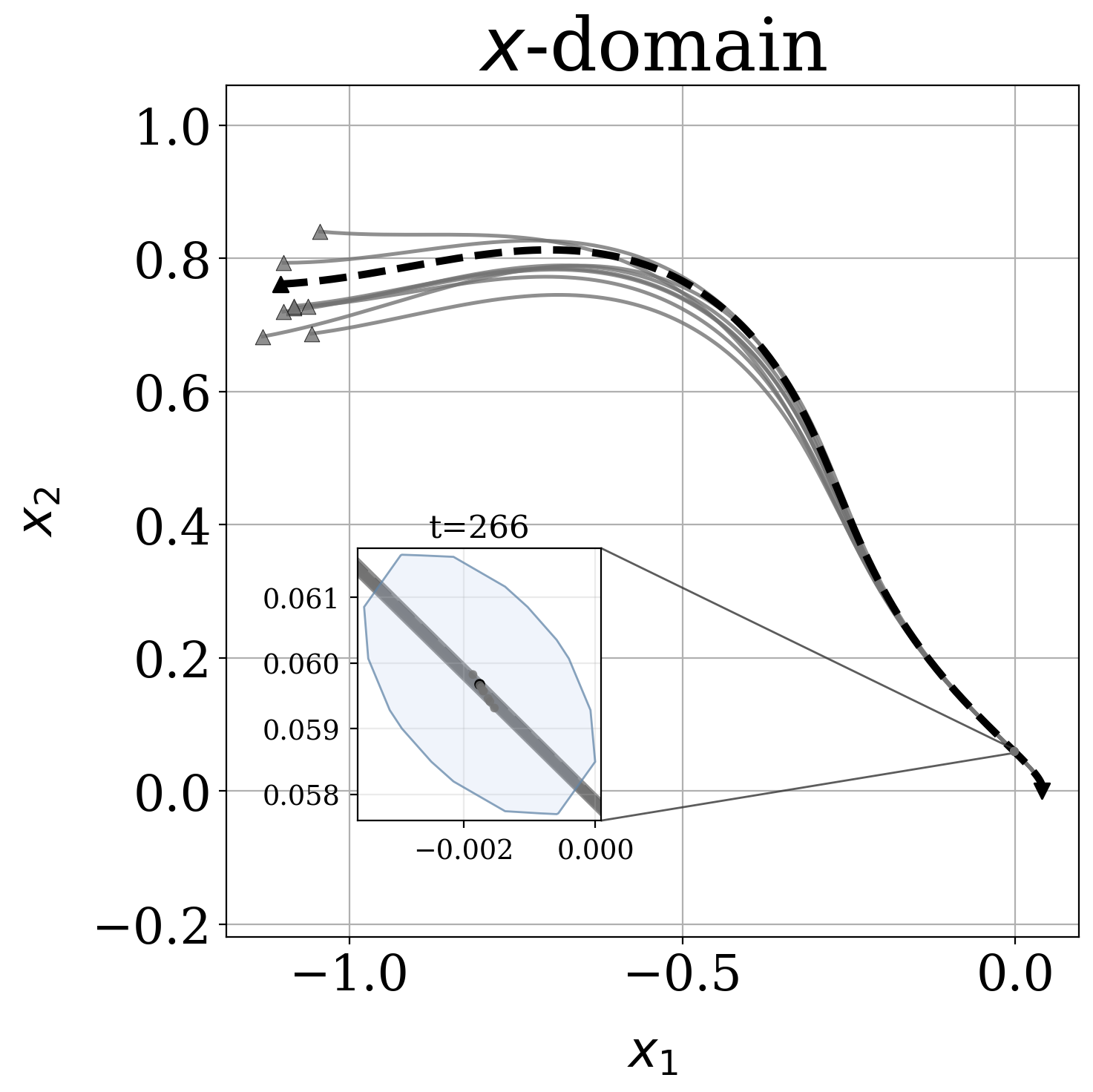}
\end{subfigure}
\begin{subfigure}{0.23\textwidth}
    \centering
    \includegraphics[width=\linewidth, keepaspectratio]{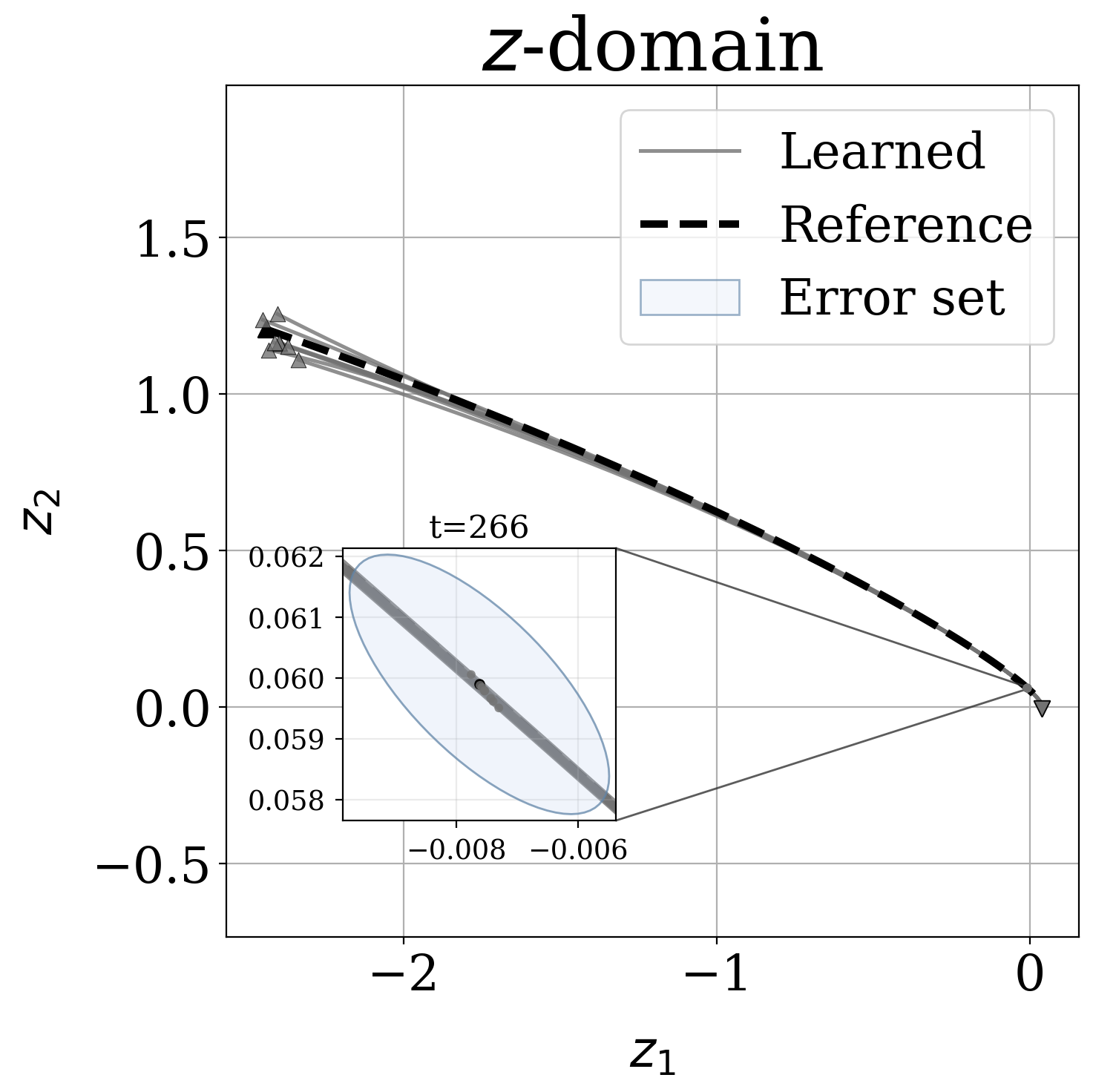}
\end{subfigure}
\caption{Learned controller tracking a reference trajectory from perturbed initial states. All trajectories converge to the reference and enter the estimated forward-invariant error tube (Theorem 2), shown in the insets, where they subsequently remain. Learning-error bounds are estimated via sampling.
}
\label{fig:trajectories}
\end{figure}
\begin{comment}
\begin{figure}[h]
\centering
\begin{subfigure}{0.2\textwidth}
    \centering
    \includegraphics[width=\linewidth, keepaspectratio]{assets/synthetic_2.png}
    \caption{}
\end{subfigure}
    \hspace{5mm}
\begin{subfigure}{0.2\textwidth}
    \centering
    \includegraphics[width=\linewidth, keepaspectratio,]{assets/synthetic_1.png}
    \caption{}
\end{subfigure}
\caption{Representative trajectories for a 2D synthetic SISO system. Markers denote the initial states of trajectories.}
\label{fig:trajectories}
\end{figure}
\end{comment}
\begin{figure}[h]
\centering
\begin{subfigure}{0.23\textwidth}
    \centering
    \includegraphics[width=\linewidth, keepaspectratio]{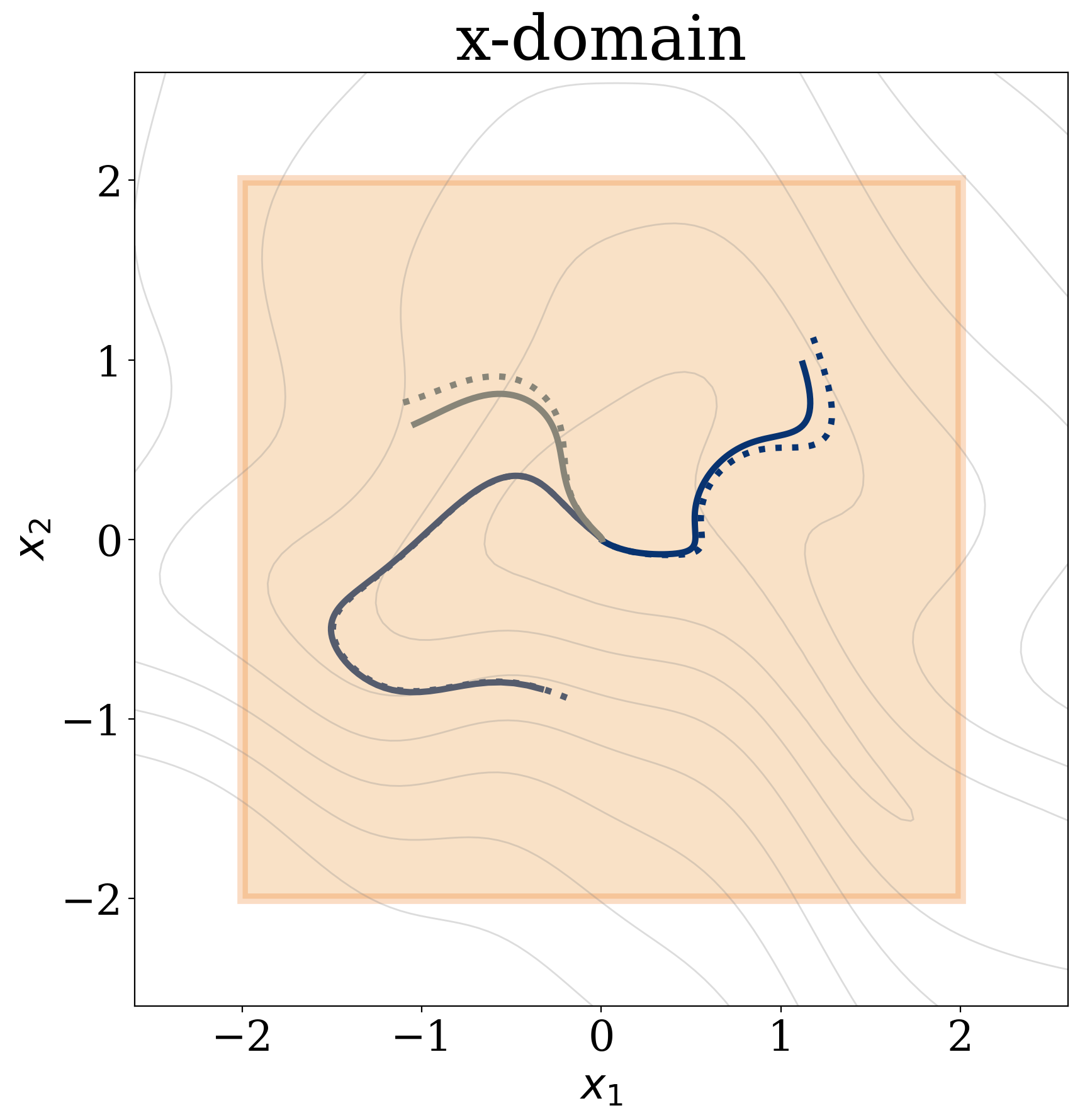}
\end{subfigure}
\begin{subfigure}{0.23\textwidth}
    \centering
    \includegraphics[width=\linewidth, keepaspectratio]{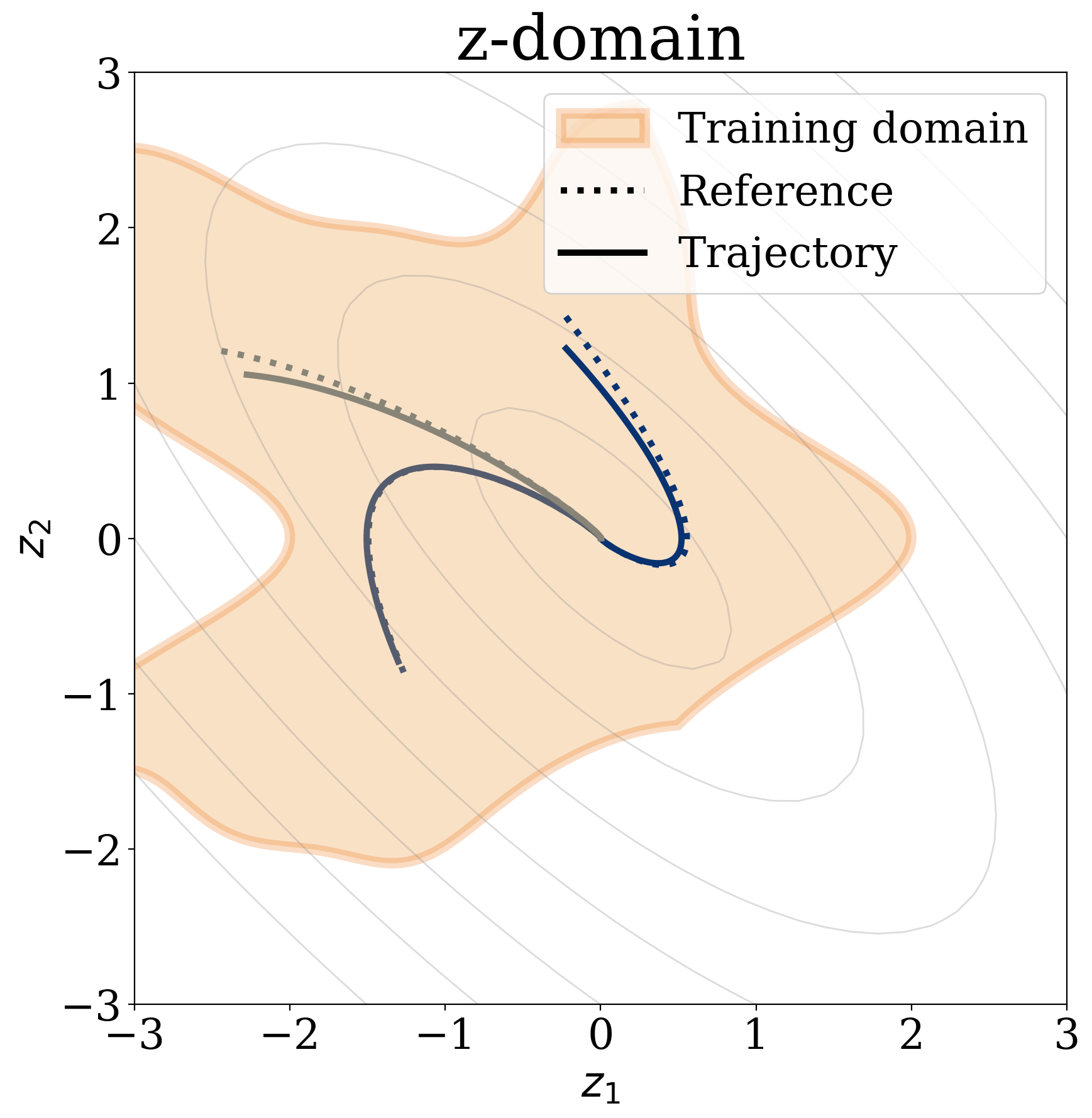}
\end{subfigure}
\caption{Sample trajectories for a 2D synthetic SISO system in both actual $\vcx$-space and linearized $\vcz$-space.}
\label{fig:xz-comparison}
\end{figure}

\subsubsection{4D SISO Systems}\label{sec:4dsiso} 
We now compare the direct and cascadic methods on 4D SISO systems. Across 10 randomly generated systems, the cascadic approach achieves an RLQE two orders of magnitude lower than the direct method and prevents the severe folding observed in the latter, effectively mitigating problems arising from high-order derivatives supervision. Furthermore, we empirically observe that the quasi-singularities of $D_\vcx\Phi$ are highly correlated with large evaluation errors of the supervision losses across all methods and systems. Finally, our cascadic approach is consistently stable and matches the tracking performance of the ground-truth controller. In contrast, the direct approach diverges on \textit{half} of the tested systems (29 out of 640 trajectories across 10 systems). These results establish that a cascaded structure is critical for applying PINN-based methods to complex nonlinear systems.

\subsubsection{8D MIMO Systems}
Finally, we evaluate the generalization of both approaches to MIMO systems by considering 8D MIMO systems with two fourth-degree linearizing outputs. The increased system dimension, together with the need to learn multiple outputs, introduces additional challenges for the proposed method. However, the learned coordinate transforms again achieve tracking performance on par with the ground-truth controller. Comparing the two learning approaches, we should also emphasize that the direct approach is not only more computationally inefficient in this setting, but also yields a more fragile approximation of the diffeomorphism. Concretely, the significant empirical folded fraction implies an ill-conditioned representation, undermining the reliability of the {downstream controller}.

%% file: sections/7-conclusion.tex
\section{Conclusion}\label{sec:conclusion}
This paper introduced a cascadic physics-informed neural network method to systematically discover feedback-linearizing coordinate transforms for nonlinear systems. By decomposing complex, high-order partial differential equations into a sequence of tractable, first-order residual minimizations, our approach successfully mitigated compounding approximation errors and enabled robust downstream tracking control. Extensive validation demonstrated strong empirical performance across real robotic systems and complex synthetic systems of up to eight dimensions. Future work will extend this framework to approximately linearizable systems (\cite{krener_approximate_1984, johansen2000computational}), systems with non-trivial zero-dynamics,  and to systems where the underlying dynamics are subject to uncertainty.

%% file: sections/A1-MIMO.tex
\subsection{FBL for Multi-Input Multi-Output Systems}\label{appendix:mimo}
Similar to the SISO setting presented in Section~\ref{sec:FBL}, the conditions for exact feedback linearization of the control affine system \eqref{eq:control-affine-sys} in the MIMO setting mirrors can be characterized in terms of the relative degree of the system. Recall that a function $h(x)$ has relative degree $r$, when, for all $x\in \stD$, 
\begin{gather}
    \mathcal{L}_{g_j}\mathcal{L}_{f}^i h(\vcx)=0 \quad\forall i\in\{0,\ldots,r-2\}, j\in\{1,\ldots,m\},\\
    \mathcal{L}_{g_j}\mathcal{L}_{f}^{r-1} h(\vcx)\neq0 \quad \forall j\in\{1,\ldots,m\}.
\end{gather}
Consider a set of real-valued smooth functions $\{h_i(x)\}_{i=1}^m$, each with relative degree $r_i$. Denoting $\vcz_i=h_i(\vcx)$, we have that the linearized state $\vcz := \begin{pmatrix} \vcz_1,\hdots,\vcz_m \end{pmatrix}$ evolves as $m$ decoupled chains of integrators (i.e. in Brunovsky canonical form). Specifically, denoting the $r$-th time derivative of $z_i$ as $h^{(r)}$, we get:
\begin{equation}
    \begin{pmatrix}
    \vcz_1^{(r_1)},\hdots,\vcz_m^{(r_m)}
    \end{pmatrix}= \vcv := b(\vcx) + A(\vcx)u,
\end{equation}
where the drift term $b:\R^n\to\R^m$ and the decoupling matrix $A:\R^n\to\R^{m\times m}$ can be written down analytically, respectively as
\begin{align}\label{eq:mimo-decoupling-matrix}
A(\vcx)=\begin{pmatrix}
    \mathcal{L}_{g_1}\mathcal{L}_{f}^{r_1-1}h_1(\vcx)&\hdots& \mathcal{L}_{g_m}\mathcal{L}_{f}^{r_1-1}h_1(\vcx) \\
    \vdots&\ddots&\vdots\\
    \mathcal{L}_{g_1}\mathcal{L}_{f}^{r_m-1}h_m(\vcx)&\hdots& \mathcal{L}_{g_m}\mathcal{L}_{f}^{r_m-1}h_m(\vcx)
\end{pmatrix},\end{align}
\vspace{0.4mm}
\begin{equation}\label{eq:mimo-drift-term}
b(\vcx)=\begin{pmatrix}
\lie_f^{r_1}h_1(\vcx), \hdots, \lie_f^{r_m}h_m(\vcx)
\end{pmatrix}.
\end{equation}

Denote the map $\Phi:\R^n\to\R^n$ as:
\begin{equation}\label{eq:mimo-diffeo}
   \Phi(\vcx)=\begin{pmatrix}
    h_1(\vcx),\hdots,\lie_f^{r_1-1}h_1(\vcx),\hdots, h_m(\vcx),\hdots,\lie_f^{r_m-1}h_m(\vcx) 
\end{pmatrix}.
\end{equation}
When the decoupling matrix $A(\vcx)$ is invertible and the total relative degree of the system $r := r_1+\hdots+r_m$ matches the state dimension $n$, i.e.,
\begin{equation}
    r_1+\hdots+r_m = n,
\end{equation}
the map $\Phi$ is a diffeomorphism and a linearizing coordinate transform for the system \eqref{eq:control-affine-sys}. The corresponding linearizing control law is given as
\begin{equation}\label{eq:mimo-control-law-apdx}
    \vcu = A(\vcx)^{-1}(\vcv - b(\vcx)),
\end{equation}
which is equivalent to the SISO case.

%% file: sections/A2-error-bounds.tex
\subsection{Proof of Theorem~\ref{thm:tracking} in the General MIMO setting}\label{appendix:proof}
We prove Theorem~\ref{thm:tracking} in the MIMO setting, which implies the results in Section~\ref{sec:controller} as special cases. Similar to the SISO setting, we consider the task of following a feasible trajectory $(\bar\vcx, \bar\vcu)$ and assume that the reference trajectory lies within the training domain.
\begin{cob}{}
\begin{assumption}\label{assm:mimo-bounded-reference}
    For all $t \geq 0$, the reference trajectory lies within the domain, i.e.,
    $ \bar\vcx(t) \in \stD$, and the reference inputs $\bar \vcv$ are uniformly bounded: $\norm{\bar \vcv(t)} \leq \bar V$.
\end{assumption}
\end{cob}

Recall that the linearizing coordinate transform takes the form of \eqref{eq:mimo-diffeo}. Therefore, we denote our in learned, \textit{approximate} diffeomorphism $\hat\Phi$ as
\begin{equation*}
    \vcz
    = \hat \Phi(\vcx)
    = \Big( \hat{\Phi}_{1,1}(\vcx), \cdots, \hat{\Phi}_{1, r_1}(\vcx), \cdots, \hat{\Phi}_{m,1}(\vcx), \cdots, \hat{\Phi}_{m, r_m}(\vcx) \Big)^\top.
\end{equation*}
Mirroring the expressions in \eqref{eq:mimo-decoupling-matrix} and \eqref{eq:mimo-drift-term}, we define
\begin{equation*}
    A_{\hat\Phi}(\vcx) = \begin{pmatrix}
        \mathcal{L}_{g_1}\hat{\Phi}_{1,r_1}(\vcx)&\hdots& \mathcal{L}_{g_m}\hat{\Phi}_{1,r_1}(\vcx) \\
        \vdots&\ddots&\vdots\\
        \mathcal{L}_{g_1}\hat{\Phi}_{m,r_m}(\vcx)&\hdots& \mathcal{L}_{g_m}\hat{\Phi}_{m,r_m}(\vcx)
    \end{pmatrix},\quad
    b_{\hat\Phi}(\vcx) = \begin{pmatrix}
        \lie_f \hat{\Phi}_{1,r_1}(\vcx)\\ \vdots\\ \lie_f \hat{\Phi}_{m,r_m}(\vcx)
    \end{pmatrix}.
\end{equation*}
The error-feedback control law takes the same form as in the SISO setting:
\begin{equation}\label{eq:mimo-control-law}
    \vcu(\vcx, t) = -A_{\hat{\Phi}}^{-1}(\vcx) \left( b_{\hat{\Phi}}(\vcx) - (\bar \vcv(t) + \mtK(\vcz - \bar \vcz(t))\right).
\end{equation}
Here, $\bar\vcz:=\hat\Phi(\vcx)$, $\bar\vcv:=A_{\hat{\Phi}}(\vcx)u + b_{\hat{\Phi}}(\vcx)$, and the robust feedback gain $\mtK$ is found by solving the $\stH_\infty$ Riccati equation \eqref{eq:h-infty-riccati} with appropriate choices of $\mtQ$ and $\mtR$.

To start, we note that under the coordinate transform $\hat\Phi$, the \textit{approximately linearized} $\vcz$ dynamics follows
\begin{equation}\label{eq:mimo-z-dynamics}
    \dot\vcz
    = \frac{d}{dt} \hat {\hat\Phi}(\vcx)
    = \begin{bmatrix}
        \frac{d}{dt} {\hat\Phi}_{1,1}(\vcx) \\ \vdots \\ \frac{d}{dt}{\hat\Phi}_{1, r_1 - 1}(\vcx) \\ \vdots
    \end{bmatrix}
    = \begin{bmatrix}
        \lie_f {\hat\Phi}_{1,1}(\vcx) \\ \vdots \\ \lie_f {\hat\Phi}_{1,r_1 - 1}(\vcx) \\ \vdots
    \end{bmatrix}
    + \begin{bmatrix}
        \sum_{i=1}^{m}\lie_{g_i} {\hat\Phi}_{1,1}(\vcx) \vcu_i \\ \vdots \\  \sum_{i=1}^{m}\lie_{g_i} {\hat\Phi}_{1,r_1 - 1}(\vcx) \vcu_i \\ \vdots
    \end{bmatrix}.
\end{equation}
To streamline notation, we define the vector function $\vcd^{c}: \stD \to \R^{n}$ and matrix function $\stR: \stD \to \R^{n \times m}$ as
\begin{equation*}
    \vcd^c(\vcx) := \begin{bmatrix}
        \lie_f {\hat\Phi}_{1,1}(\vcx) - {\hat\Phi}_{1,2}(\vcx) \\ \vdots \\ \lie_f {\hat\Phi}_{1,r_1 - 1}(\vcx) - {\hat\Phi}_{1,r_1}(\vcx)\\ \vdots
    \end{bmatrix},\quad
    [\stR(\vcx)]_{i,j} = \lie_{g_j}{\hat\Phi}_i(\vcx).
\end{equation*}
Applying the control \eqref{eq:mimo-control-law}, the $\vcz$ dynamics \eqref{eq:mimo-z-dynamics} evolves as
\begin{equation}
\begin{aligned}
    \dot\vcz
    &= \left(\mtA \vcz + \mtB \vcv + \vcd^c(\vcx)\right) + \stR(\vcx) \vcu \\
    &= \left(\mtA \vcz + \mtB \vcv + \vcd^c(\vcx)\right) + \stR(\vcx) A_{\hat\Phi}^{-1}(\vcx) \left(\vcv - b_{\hat\Phi}(\vcx)\right)\\
    &= \mtA \vcz + \left(\mtB + \underbrace{\stR(\vcx) A^{-1}(\vcx)}_{=:\Delta_B(\vcx)} \right) \vcv + \left(\underbrace{\vcd^c(\vcx) - \stR(\vcx) A_{\hat\Phi}^{-1}(\vcx) b_{\hat\Phi}(\vcx)}_{=:\vcd^z(\vcx)}\right).
\end{aligned}
\end{equation}
\noindent Define the $\vcz$-state error as $\vce := \vcz - \bar\vcz$. The error dynamics evolve as
\begin{equation*}
\begin{aligned}
    \dot\vce
    % &= \mtA(\vcz-\bar\vcz) + \mtB (\cancel{\bar v} + K \vce - \cancel{\bar v}) + \mtDelta_B(\vcx) (\bar v + K \vce) + \vcd^z(\vcx) - \mtDelta_B(\bar\vcx)\bar v - \vcd^z(\bar\vcx)\\
    &= \mtA \vce + (\mtB + \mtDelta_B(\vcx)) \mtK\vce + \underbrace{(\mtDelta_B(\vcx) - \mtDelta_B(\bar\vcx))\bar \vcv + \vcd^z(\vcx) - \vcd^z(\bar\vcx)}_{=:\vcd}.
\end{aligned}
\end{equation*}
We assume that the error terms $\Delta_B$ and $\vcd^z$ are small when the models are well-trained on the domain.
\begin{cob}{}
\begin{assumption}\label{assm:mimo-bounded-residual}
    We assume that for all $\vcx \in \stD$, the error terms $\Delta_B$ and $\vcd^z$ are bounded:
    \begin{equation*}
        \norm{\Delta_B(\vcx)}_{2\to 2} \leq \bar\delta,\quad
        \norm{\vcd^z(\vcx)}_2 \leq \bar d^z.
    \end{equation*}
\end{assumption}
\end{cob}
\noindent Note that from Assumption~\ref{assm:mimo-bounded-residual}, we clearly have that
\begin{equation*}
    \norm{\vcd}
    \leq \norm{\Delta_B(\vcx) + \Delta_B(\bar\vcx)}\norm{\bar\vcv} + \norm{\vcd^z(\vcx) - \vcd^z(\bar\vcz)}
    \leq 2(\bar\delta\bar V + \bar d^z) =: \bar d.
\end{equation*}
The boundedness of the approximation errors and the reference signal allows us to establish a forward-invariant set of the $\vcz$ tracking error.

\begin{cob}{}
    \begin{lemma}\label{lem:z-error}
        Under Assumptions~\ref{assm:mimo-bounded-reference} and \ref{assm:mimo-bounded-residual}, if $\bar\delta < \frac{\sqrt{\lmin(\mtR)}}{\gamma}$ and if $\vcx(t) \in \stD$ for all $t\geq 0$, the tracking error $\vce$ is uniformly ultimately bounded by the sublevel set
        \begin{equation}\label{eq:mimo-z-sublevel-set}
            \Omega_z = \left\{\vce \mid \vce^\top \mtP \vce \leq \frac{\lmax(\mtP)\bar d^2}{\lmin(\mtQ)\left( \frac{1}{\gamma^2} - \frac{\bar\delta^2}{\lmin(\mtR)} \right)}\right\}.
        \end{equation}
    \end{lemma}
\end{cob}

\noindent \textit{Proof.} Consider the Lyapunov function for the error dynamics
\begin{equation}
    V(\vce) = \frac{1}{2}\vce^\top \mtP\vce,
\end{equation}
we have that
\begin{equation}
\begin{aligned}
    \dot V(\vce)
    &= \frac{1}{2}\left[\vce^\top \mtP (\mtA + (\mtB + \Delta_B)\mtK\vce + \vcd) + (\mtA + (\mtB + \Delta_B)\mtK\vce + \vcd)^\top \mtP \vce\right] \\
    &= \frac{1}{2}\vce^\top(\mtP\mtA + \mtA^\top \mtP - 2\mtP\mtB\mtR^{-1}\mtB^\top \mtP)\vce + \frac{1}{2}\vce^\top(\mtP\Delta_\mtB\mtK + \mtK^\top \Delta_B^\top \mtP)\vce + \vce^\top \mtP \vcd\\
    &= \frac{1}{2}\vce^\top(-\mtQ-\frac{\mtP\mtP}{\gamma^2} - \mtP\mtB\mtR^{-1}\mtB^\top \mtP)\vce -\frac{1}{2}\vce^\top(\mtP\Delta_\mtB\mtR^{-1}\mtB^\top \mtP + \mtP\mtB\mtR^{-1} \Delta_B^\top \mtP)\vce + \vce^\top \mtP \vcd\\
    &= -\frac{1}{2}\vce^\top \mtQ \vce - \frac{1}{2\gamma^2}\vce^\top \mtP \mtP\vce - \frac{1}{2}\underbrace{\vce^\top\left(\mtP\mtB\mtR^{-1}\mtB^\top \mtP + \mtP\Delta_B\mtR^{-1}\mtB^\top \mtP + \mtP\mtB\mtR^{-1}\Delta_B^\top \mtP \right)\vce}_{(\triangle)} + \vce^\top \mtP \vcd.
\end{aligned}
\end{equation}
Denoting $\vcw:= \mtR^{-1}\mtB^\top \mtP\vce$,
\begin{equation*}
\begin{aligned}
    (\triangle) &= \vcw^\top \mtR \vcw + (\mtP \vce)^\top \Delta_B \vcw + \vcw^\top \Delta_B^\top \mtP\vce\\
    &= (\vcw + \mtR^{-1} \Delta_B^\top \mtP\vce)^\top \mtR (\vcw + \mtR^{-1} \Delta_B^\top \mtP\vce) - (\mtP\vce)^\top \Delta_B \mtR^{-1} \Delta_B^\top \mtP \vce\\
    &\geq - (\mtP\vce)^\top \Delta_B \mtR^{-1} \Delta_B^\top \mtP \vce.
\end{aligned}
\end{equation*}
As a result, we can bound $\dot V$ as
\begin{equation*}
\begin{aligned}
    \dot V(\vce)
    &\leq -\frac{1}{2}\vce^\top \mtQ \vce - \frac{1}{2\gamma^2}\vce^\top \mtP \mtP\vce + \frac{1}{2}(\mtP\vce)^\top \Delta_B \mtR^{-1} \Delta_B^\top \mtP \vce + \vce^\top \mtP \vcd\\
    &= -\frac{1}{2}\vce^\top \mtQ \vce - \frac{1}{2}\vce^\top \mtP \left(\frac{1}{\gamma^2}I - \Delta_B \mtR^{-1} \Delta_B^\top \right) \mtP \vce + \vcd^\top \mtP \vce\\
    &\leq -\frac{1}{2}\vce^\top \mtQ \vce - \frac{1}{2}\lmin\left(\frac{1}{\gamma^2}I - \Delta_B \mtR^{-1} \Delta_B^\top\right)\norm{\mtP\vce}^2 + \bar d \norm{\mtP\vce}
\end{aligned}
\end{equation*}
The relevant minimum eigenvalue is bounded as
\begin{equation*}
    \lmin\left(\frac{1}{\gamma^2}I - \Delta_B \mtR^{-1} \Delta_B^\top\right)
    =\frac{1}{\gamma^2} - \lmax\left(\Delta_B \mtR^{-1} \Delta_B^\top\right)
    \geq \underbrace{\frac{1}{\gamma^2} - \frac{\bar\delta^2}{\lmin(\mtR)}}_{=:\eta} > 0,
\end{equation*}
where the last inequality follows from our assumption that $\bar\delta < \sqrt{\lmin(\mtR)} / \gamma$. Denoting this lower bound as $\eta$, completing the square on $\norm{\mtP\vce}$ yields
\begin{equation*}
\begin{aligned}
    \dot V(\vce)
    &\leq -\frac{1}{2}\vce^\top \mtQ \vce - \frac{1}{2}\eta\norm{\mtP\vce}^2 + \bar d \norm{\mtP\vce}\\
    &\leq -\frac{1}{2}\vce^\top \mtQ \vce - \frac{1}{2}\left(\sqrt{\eta}\norm{\mtP\vce} - \frac{\bar d}{\sqrt{\eta}}\right)^2 + \frac{\bar d^2}{2\eta}\\
    &\leq -\frac{1}{2}\lmin(\mtQ)\norm{\vce}^2 + \frac{\bar d^2}{2\eta} \\
\end{aligned}
\end{equation*}
As a result, the Lyapunov function is strictly decreasing when $\norm{\vce}^2 > \frac{\bar d^2}{\lmin(\mtQ)\eta}$. The tracking error is thus uniformly ultimately bounded by the sublevel set \eqref{eq:mimo-z-sublevel-set}.\qed

To connect the error of linearized states $\vcz$, to that of the physical states $\vcx$, we leverage Assumption~\ref{assm:lipschitz-T}, which states that the inverse of the diffeomorphism ${\hat\Phi}^{-1}$ is well-conditioned. This leads to the following result.

\begin{cob}{}
    \begin{lemma}\label{lem:x-error}
    Under the assumptions of Lemma~\ref{lem:z-error} and Assumption~\ref{assm:lipschitz-T}, the $\vcx$ tracking error $\vce_x := \vcx - \bar\vcx$ is uniformly ultimately bounded by the set
    \begin{equation*}
        \Omega_x = \left\{\vce \mid \norm{\vce_x}^2 \leq \frac{\kappa(\mtP)\bar d^2}{l_{\hat\Phi}^2\lmin(\mtQ)\left( \frac{1}{\gamma^2} - \frac{\bar\delta^2}{\lmin(\mtR)} \right)}\right\},
    \end{equation*}
    where $\kappa(\cdot)$ denotes the condition number of the given matrix.
    \end{lemma}
\end{cob}
\noindent \textit{Proof.} From Lemma~\ref{lem:z-error}, there exists a time $T$, independent of starting time $t_0$, such that for $t \geq T$,
\begin{equation*}
    (\vcz - \bar\vcz)^\top \mtP (\vcz - \bar\vcz) \leq \frac{\lmax(\mtP)\bar d^2}{\lmin(\mtQ)\left( \frac{1}{\gamma^2} - \frac{\bar\delta^2}{\lmin(\mtR)} \right)}.
\end{equation*}
Thus, for the same time $T$, we have that
\begin{equation*}
\begin{aligned}
    \norm{\vce_x}^2
    &= \norm{\vcx - \bar\vcx}^2 \\
    &\leq \frac{1}{l_{\hat\Phi}^2} \norm{{\hat\Phi}(\vcx) - {\hat\Phi}(\bar\vcx)}^2 \\
    % &\leq \frac{1}{l_{\hat\Phi}^2 \lmin(P)} \lmin(P) \norm{\vcz - \bar\vcz}^2 \\
    &\leq \frac{1}{l_{\hat\Phi}^2 \lmin(\mtP)} \vce_z^\top \mtP \vce_z \\
    &\leq \frac{1}{l_{\hat\Phi}^2 \lmin(\mtP)} \frac{\lmax(\mtP)\bar d^2}{\lmin(\mtQ)\left( \frac{1}{\gamma^2} - \frac{\bar\delta^2}{\lmin(\mtR)} \right)}.
\end{aligned}
\end{equation*}
\qed

Finally, the above results explicitly assume that $\vcx$ does not leave the domain $\stD$, where the approximation quality is guaranteed. We thus leverage the forward invariant error set to robustify our assumption on our reference trajectory.
\begin{cob}{}
\begin{assumption}\label{assm:robust-reference}
    We assume that the reference trajectory states lie within the domain, i.e.,
    $$ \bar\vcx(t) \in \stD \ominus \Omega_x,\quad \forall t \geq 0,$$
    and that the reference inputs $\bar v$ is uniformly bounded:
    $$ \norm{\bar v(t)} \leq \bar V,\quad \forall t \geq 0.$$
    Here $\ominus$ denotes the Minkowski difference.
\end{assumption}
\end{cob}
\begin{theorem}\label{thm:tracking_appendix}
    Under assumptions~\ref{assm:mimo-bounded-residual}, \ref{assm:lipschitz-T}, and \ref{assm:robust-reference}, if $\bar\delta < \frac{\sqrt{\lmin(\mtR)}}{\gamma}$ and $\vce_x(0) \in \Omega_x$, the tracking error $\vce_x$ will be uniformly ultimately bounded by the sublevel set $\Omega_x$.
\end{theorem}
\noindent \textit{Proof.} Following the assumptions, the state trajectory always lies within $\stD$. The rest follows directly from Lemma~\ref{lem:x-error}. \qed

%% file: sections/A3-on-the-diffeomorphism.tex
\subsection{A note on the diffeomorphic property of the learned transform}

\paragraph{Exact feedback linearization.}
As already established in Section~\ref{sec:FBL}, the existence of an observable function $h(\vcx)$ that satisfies conditions \textbf{CA} (or, equivalently, \textbf{CB}) generates a coordinate transformation $\Phi(\vcx)$ that feedback linearizes the system by stacking the iterated Lie derivatives of $h$. While conditions \textbf{CA}, together with successive time derivatives of $h(\vcx)$ clearly demonstrate that the transformed dynamics form a chain of integrators, showing that $\Phi(\vcx)$ constitutes a valid coordinate transformation requires establishing the nonsingularity of its Jacobian.

Following Lemma~4.1.1 in~\cite{isidori1985nonlinear}, consider the matrix-valued function $\mtM:\R^{n}\to\R^{n\times n}$ defined as
\begin{equation}
\mtM
=
D_\vcx\Phi\,
\begin{bmatrix}
g & \adfg{} & \hdots & \adfg{n-1}
\end{bmatrix}.
\label{eq:M_definition}
\end{equation}
Under the standing assumptions for state feedback linearizability, the bracket matrix
\begin{equation}
\mtG(\vcx)
=
\begin{bmatrix}
g & \adfg{} & \hdots & \adfg{n-1}
\end{bmatrix}
\end{equation}
is nonsingular on $\mathcal{D}$. Therefore, nonsingularity of $\mtM$ implies nonsingularity of $D_\vcx\Phi=\mtM\mtG^{-1}$.

Invoking the basic algebraic properties of Lie brackets, $\mtM$ can be constructed column by column through the recursive relation (see Lemma~4.1.2 in~\cite{isidori1985nonlinear})
\begin{equation}
\mtM_{i,j+1}
=
\lie_f\mtM_{i,j}
-
\mtM_{i+1,j},
\qquad
i=1,\ldots,n-1.
\label{eq:M_exact_recursion}
\end{equation}
In addition, by construction of the feedback-linearizing coordinates,
\begin{equation}
\phi_i=\lie_f^{\,i-1}h,
\end{equation}
and the relative-degree-$n$ condition gives
\begin{equation}
\lie_g h
=
\lie_g\lie_f h
=
\cdots
=
\lie_g\lie_f^{n-2}h
=
0,
\qquad
\beta
:=
\lie_g\lie_f^{n-1}h
\neq 0.
\end{equation}
Consequently, the first column of $\mtM$ is
\begin{equation}
\mtM_{:,0}
=
\begin{bmatrix}
0 & 0 & \hdots & 0 & \beta
\end{bmatrix}^{\top}.
\label{eq:M_initial_exact}
\end{equation}
Propagating~\eqref{eq:M_initial_exact} through~\eqref{eq:M_exact_recursion} yields the lower anti-triangular structure
\begin{equation}
\mtM =
\begin{bmatrix}
0      & 0      & 0      & \cdots & 0      & (-1)^{n-1}\beta \\
0      & 0      & 0      & \cdots & (-1)^{n-2}\beta & * \\
0      & 0      & 0      & \iddots & *      & * \\
\vdots & \vdots & \iddots & \iddots & \vdots & \vdots \\
0      & -\beta & *      & \cdots & *      & * \\
\beta  & *      & *      & \cdots & *      & *
\end{bmatrix}.
\label{eq:M_exact_structure}
\end{equation}
The anti-diagonal entries are therefore nonzero whenever $\beta\neq0$, which is sufficient to establish the nonsingularity of $\mtM$ and, consequently, of $D_\vcx\Phi$.

\paragraph{Approximate control conditions.}
The previous argument relies on exact satisfaction of both the control conditions and the Lie-derivative cascade defining $\Phi$. Nevertheless, the same construction provides a useful lens through which to understand how approximation errors affect the diffeomorphic property. We first consider the simpler case in which the Lie-derivative cascade remains exact, while the control conditions are satisfied only approximately. Let the learned coordinates be denoted by $\hat\phi_i$, $i=1,\ldots,n$, and assume
\begin{equation}
\hat\phi_{i+1}
=
\lie_f\hat\phi_i,
\qquad
i=1,\ldots,n-1.
\label{eq:exact_learned_cascade}
\end{equation}
The deviations from the control conditions can then be collected in the first column of $\hat{\mtM}$ by defining
\begin{equation}
\epsilon_i
=
\lie_g\hat\phi_i,
\qquad
i=1,\ldots,n-1,
\qquad
\epsilon_n
=
\lie_g\hat\phi_n-\beta.
\label{eq:control_residuals}
\end{equation}
Consequently,
\begin{equation}
\hat{\mtM}_{:,0}
=
\begin{bmatrix}
\epsilon_1 &
\epsilon_2 &
\hdots &
\epsilon_{n-1} &
\beta+\epsilon_n
\end{bmatrix}^{\top}.
\label{eq:M_initial_perturbed}
\end{equation}

Since the Lie-derivative cascade is still satisfied exactly, the same recursion as in the exact case holds,
\begin{equation}
\hat{\mtM}_{i,j+1}
=
\lie_f\hat{\mtM}_{i,j}
-
\hat{\mtM}_{i+1,j}.
\end{equation}
Let $\Delta\mtM^{g}$ denote the component of the perturbation generated by the control residuals in the recursively determined portion of $\mtM$. Its propagation is governed by
\begin{equation}
\Delta\mtM^{g}_{i,j+1}
=
\lie_f\Delta\mtM^{g}_{i,j}
-
\Delta\mtM^{g}_{i+1,j},
\qquad
\Delta\mtM^{g}_{i,0}
=
\epsilon_i.
\label{eq:control_error_recursion}
\end{equation}
The first few propagated residuals are therefore
\begin{align}
\Delta\mtM^{g}_{i,1}
&=
\lie_f\epsilon_i-\epsilon_{i+1},
\\
\Delta\mtM^{g}_{i,2}
&=
\lie_f^2\epsilon_i
-
2\lie_f\epsilon_{i+1}
+
\epsilon_{i+2},
\end{align}
and, more generally,
\begin{equation}
\Delta\mtM^{g}_{i,j}
=
\sum_{k=0}^{j}
(-1)^k
{j\choose k}
\lie_f^{\,j-k}\epsilon_{i+k},
\qquad
i+j\le n.
\label{eq:control_error_closed_form}
\end{equation}
Hence, the exact zero structure of $\mtM$ is no longer preserved: the control residuals appearing in the first column propagate toward the upper anti-triangular region and the anti-diagonal of $\mtM$. Importantly, uniformly small values of $\epsilon_i$ do not directly imply a uniformly small perturbation of $\mtM$, since the propagated entries depend on increasingly high Lie derivatives of the residual functions.

\paragraph{Approximate Lie-derivative consistency.}
A second source of approximation arises when the learned coordinates do not satisfy the exact Lie-derivative structure, as is the case in our cascadic implementation. Define the corresponding consistency residuals as
\begin{equation}
c_i
=
\lie_f\hat\phi_i-\hat\phi_{i+1},
\qquad
i=1,\ldots,n-1,
\label{eq:consistency_error}
\end{equation}
so that
\begin{equation}
\lie_f\hat\phi_i
=
\hat\phi_{i+1}+c_i.
\end{equation}
Letting
\begin{equation}
V_j=\adfg{j},
\end{equation}
the entries of the learned matrix are
\begin{equation}
\hat{\mtM}_{i,j}
=
\lie_{V_j}\hat\phi_i.
\end{equation}
Using $V_{j+1}=[f,V_j]$ and the commutator identity
\begin{equation}
\lie_{[f,V_j]}
=
\lie_f\lie_{V_j}
-
\lie_{V_j}\lie_f,
\end{equation}
the recursive construction becomes
\begin{equation}
\hat{\mtM}_{i,j+1}
=
\lie_f\hat{\mtM}_{i,j}
-
\hat{\mtM}_{i+1,j}
-
\lie_{\adfg{j}}c_i.
\label{eq:full_error_recursion}
\end{equation}
Thus, while the control residuals modify the initial condition of the recursion through the first column, the consistency residuals act as an additional forcing term injected at each stage of the recursive construction.

Subtracting the exact recursion~\eqref{eq:M_exact_recursion} gives the perturbation dynamics
\begin{equation}
\Delta\mtM_{i,j+1}
=
\lie_f\Delta\mtM_{i,j}
-
\Delta\mtM_{i+1,j}
-
\lie_{\adfg{j}}c_i,
\qquad
\Delta\mtM_{i,0}
=
\epsilon_i.
\label{eq:full_perturbation_recursion}
\end{equation}
Accordingly, the total perturbation can be interpreted as the superposition of a \emph{free response}, generated by the control residuals $\epsilon_i$ through the initial condition, and a \emph{forced response}, generated by the consistency residuals $c_i$. Introducing the propagation operator
\begin{equation}
(\mathcal{T}z)_i
=
\lie_f z_i-z_{i+1},
\end{equation}
and the forcing
\begin{equation}
(q_j)_i
=
\lie_{\adfg{j}}c_i,
\end{equation}
the perturbation in the $j$-th column can be written compactly as
\begin{equation}
\Delta\mtM_j
=
\mathcal{T}^{j}\epsilon
-
\sum_{s=0}^{j-1}
\mathcal{T}^{\,j-1-s}q_s.
\label{eq:free_forced_response}
\end{equation}
Due to the linearity of the recursive relation, this decomposition is directly analogous to the decomposition of a linear discrete system into its zero-input and zero-state responses. Both contributions involve derivatives of the learned residual functions and may therefore differ spatially from the residual values directly penalized during training. Importantly, however, the consistency contribution does not introduce an \emph{a priori} higher differential order than the control contribution: at column $j$, both involve derivatives of their respective residuals of order at most $j$, with consistency errors introduced at later stages (which empirically appear larger) undergoing fewer subsequent propagation steps.

\paragraph{Structured perturbation and preservation of invertibility.}
To study how the resulting perturbation affects invertibility, it is convenient to reverse the column ordering of $\mtM$ through a fixed permutation matrix $\mtP$. The exact matrix
\begin{equation}
    \mtL^\star = \mtM\mtP
\end{equation}
is lower triangular with nonzero diagonal and is therefore invertible. The learned matrix can be decomposed as
\begin{equation}
    \hat{\mtM}\mtP
    =
    \mtL^\star
    +
    \Delta\mtL
    +
    \mtU,
    \label{eq:triangular_perturbation_decomposition}
\end{equation}
where $\Delta\mtL$ contains the perturbation in the strictly lower-triangular region, while $\mtU$ contains the complementary upper-triangular perturbation, including perturbations of the diagonal. This decomposition is applied to the total perturbation, independently of whether its entries originate from control or consistency errors.

Factoring out the exact matrix gives
\begin{equation}
    \hat{\mtM}\mtP
    =
    \mtL^\star
    \left(
    \mtI+\mtA+\mtB
    \right),
    \label{eq:AB_factorization}
\end{equation}
where
\begin{equation}
    \mtA
    =
    (\mtL^\star)^{-1}\Delta\mtL,
    \qquad
    \mtB
    =
    (\mtL^\star)^{-1}\mtU.
\end{equation}
Since $(\mtL^\star)^{-1}$ is lower triangular and $\Delta\mtL$ is strictly lower triangular, $\mtA$ is itself strictly lower triangular and hence nilpotent,
\begin{equation}
    \mtA^n=0.
    \label{eq:A_nilpotent}
\end{equation}
Consequently, $\mtI+\mtA$ is always invertible and
\begin{equation}
    (\mtI+\mtA)^{-1}
    =
    \sum_{k=0}^{n-1}(-1)^k\mtA^k .
    \label{eq:nilpotent_inverse}
\end{equation}
In particular, perturbations confined to the strictly lower-triangular region cannot, by themselves, destroy invertibility.

Using~\eqref{eq:nilpotent_inverse},~\eqref{eq:AB_factorization} can be factored further as
\begin{equation}
    \hat{\mtM}\mtP
    =
    \mtL^\star
    (\mtI+\mtA)
    \left(
    \mtI+\mtC
    \right),
    \label{eq:final_factorization}
\end{equation}
where
\begin{equation}
    \mtC
    :=
    (\mtI+\mtA)^{-1}\mtB
    =
    \sum_{k=0}^{n-1}(-1)^k\mtA^k\mtB .
    \label{eq:C_definition}
\end{equation}
Since the first two factors in~\eqref{eq:final_factorization} are invertible, preservation of invertibility reduces entirely to the last factor. In particular,
\begin{equation}
    \|\mtC\|_2<1
    \label{eq:C_condition}
\end{equation}
is sufficient to guarantee that $\mtI+\mtC$ is invertible and, consequently, that $\hat{\mtM}\mtP$ remains invertible.

The triangular structure provides a useful refinement of this condition. From~\eqref{eq:C_definition},
\begin{equation}
    \mtC
    =
    \mtB
    +
    \underbrace{
    \sum_{k=1}^{n-1}(-1)^k\mtA^k\mtB
    }_{\mtR},
    \label{eq:C_remainder}
\end{equation}
where the interaction remainder $\mtR$ satisfies the exact bound
\begin{equation}
    \|\mtR\|_2
    \le
    \left(
    \sum_{k=1}^{n-1}\|\mtA\|_2^k
    \right)
    \|\mtB\|_2.
    \label{eq:R_bound}
\end{equation}
Hence,
\begin{equation}
    \|\mtC\|_2
    \le
    \left(
    \sum_{k=0}^{n-1}\|\mtA\|_2^k
    \right)
    \|\mtB\|_2,
    \label{eq:C_bound}
\end{equation}
and a sufficient condition for preservation of invertibility is therefore
\begin{equation}
    \boxed{
    \left(
    \sum_{k=0}^{n-1}\|\mtA\|_2^k
    \right)
    \|\mtB\|_2
    <1 .
    }
    \label{eq:invertibility_condition}
\end{equation}

The role of the triangular decomposition becomes particularly transparent by
normalizing the perturbation with respect to the smallest singular value of the
reference matrix. Let
\begin{equation}
    \mtE = \Delta\mtL + \mtU,
\end{equation}
where $\Delta\mtL$ and $\mtU$ denote, respectively, the strictly lower- and
upper-triangular components of the perturbation. Since their supports are
disjoint, they are orthogonal with respect to the Frobenius inner product, and
thus
\begin{equation}
    \|\mtE\|_F^2
    =
    \|\Delta\mtL\|_F^2
    +
    \|\mtU\|_F^2.
\end{equation}
In particular,
\begin{equation}
    \|\Delta\mtL\|_F \le \|\mtE\|_F,
    \qquad
    \|\mtU\|_F \le \|\mtE\|_F.
\end{equation}

Recalling
\begin{equation}
    \mtA=(\mtL^\star)^{-1}\Delta\mtL,
    \qquad
    \mtB=(\mtL^\star)^{-1}\mtU,
\end{equation}
and using
\begin{equation}
    \|(\mtL^\star)^{-1}\|_2
    =
    \frac{1}{\sigma_{\min}(\mtL^\star)},
    \qquad
    \|\mtX\|_2\le\|\mtX\|_F,
\end{equation}
we obtain
\begin{equation}
    \|\mtA\|_2
    \le
    \frac{\|\Delta\mtL\|_F}{\sigma_{\min}(\mtL^\star)}
    \le \eta,
    \qquad
    \|\mtB\|_2
    \le
    \frac{\|\mtU\|_F}{\sigma_{\min}(\mtL^\star)}
    \le \eta,
\end{equation}
where we define the dimensionless perturbation level
\begin{equation}
    \boxed{
    \eta
    :=
    \frac{\|\mtE\|_F}
         {\sigma_{\min}(\mtL^\star)}.
    }
\end{equation}

Consequently, for $\eta<1$, the remainder term in~\eqref{eq:R_bound}
satisfies
\begin{equation}
    \|\mtR\|_2
    \le
    \|\mtB\|_2
    \sum_{k=1}^{n-1}\|\mtA\|_2^k
    \le
    \eta\sum_{k=1}^{n-1}\eta^k
    =
    \eta^2
    \frac{1-\eta^{\,n-1}}{1-\eta}.
\end{equation}
Hence,
\begin{equation}
    \boxed{
    \|\mtR\|_2
    =
    \mathcal{O}(\eta^2)
    =
    \mathcal{O}\!\left(
    \frac{\|\mtE\|_F^2}
         {\sigma_{\min}^2(\mtL^\star)}
    \right).
    }
\end{equation}
Thus,
\begin{equation}
    \mtC
    =
    \mtB+\mathcal{O}(\eta^2).
    \label{eq:first_order_C}
\end{equation}
Therefore, the upper-triangular perturbation $\mtU$ constitutes the leading-order structure-breaking contribution, whereas the perturbation $\Delta\mtL$ in the lower-triangular region affects invertibility only through products with $\mtU$, and hence only at second and higher orders.

Most importantly, whenever~\eqref{eq:invertibility_condition} is satisfied,
\begin{equation}
    \hat{\mtM}\mtP
\end{equation}
is invertible. Since $\mtP$ is a permutation matrix, this immediately implies that $\hat{\mtM}$ is invertible. Finally, under the standing feedback-linearizability assumption that
\begin{equation}
    \mtG
    =
    \begin{bmatrix}
        g & \adfg{} & \hdots & \adfg{n-1}
    \end{bmatrix}
\end{equation}
is nonsingular,
\begin{equation}
    D_\vcx\hat\Phi
    =
    \hat{\mtM}\mtG^{-1}
\end{equation}
is also nonsingular. Hence, condition~\eqref{eq:invertibility_condition} provides a sufficient perturbation bound under which the learned transformation preserves the diffeomorphic property of the exact feedback-linearizing map.

%% file: sections/A4-synthesis.tex
\subsection{System synthesis details} \label{app:synthetic}
As already described in Section~\ref{subsec:synthetic}, we synthesize  exactly feedback linearizable systems via randomly generated diffeomorphisms building on a bidirectional affine coupling mechanism. After splitting the input $\vcx$ to $\vcx_1, \vcx_2$, the diffeomorphism is defined as:
\begin{align*}
\vcy_1 = \vcx_1 + t_1(\vcx_2)\\
\vcy_2 = \vcx_2 + t_2(\vcy_1)
\end{align*}
 and $t_1$, $t_2$ are parameterized as a linear combination of Gaussian kernel functions. In our experiments, we fix $s_1(\vcx)=s_2(\vcx)=\mathbf1$ and use Gaussian kernels to parameterize $t_1$ and $t_2$. Analytically, we define each transformation as $t_i(\vcx) := \mtW_i \phi_i(\vcx)$, where $\mtW_i \in \mathbb{R}^{d\times (n/2)}$ is a weight matrix and $\phi_i(\vcx) \in \mathbb{R}^d$ is a vector of Gaussian basis functions with components $\phi_{i,k}(\vcx) := \exp\!\left(-\tfrac{1}{2}(\vcx - \vcc_{i,k})^\top \mtSigma_{i,k}^{-1}(\vcx - \vcc_{i,k})\right)$ for $k=1,\dots,d$. Here, $\vcc_{i,k} \in \mathbb{R}^{n/2}$ are the kernel centers and $\mtSigma_{i,k} \in \mathbb{R}^{(n/2)\times(n/2)}$ are diagonal covariance matrices. Thus, \(t_1(\vcx)\) and \(t_2(\vcx)\) are linear combinations of Gaussian kernels. We generate diverse invertible transforms by sampling the nonzero parameters of \(t_i(\vcx)\) uniformly. The smoothness-and, consequently, learning task difficulty—is controlled by the upper bound on the diagonal covariance entries. The exact parameters used for the synthetic experiments of Section~\ref{subsec:synthetic}
 are summarized in Tables~\ref{tab:kernels1d}, \ref{tab:kernels2d} and \ref{tab:kernels8d}.
\begin{figure}
\centering
\begin{subfigure}{0.40\textwidth}
    \centering
    \includegraphics[width=\linewidth]{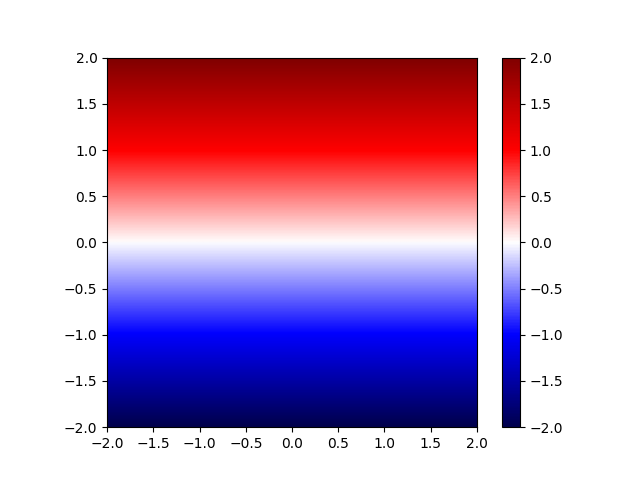}
    \caption{$h(\vcx)=\vcx_2$}
\end{subfigure}
\hfill
\begin{subfigure}{0.40\textwidth}
    \centering
    \includegraphics[width=\linewidth]{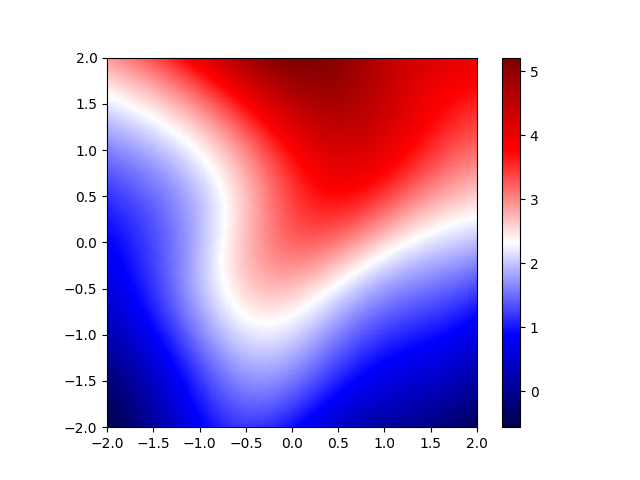}
    \caption{$\mtSigma_k \sim\mathcal{U}(0,1)$}
\end{subfigure}
\vspace{0.5em}
\begin{subfigure}{0.40\textwidth}
    \centering
    \includegraphics[width=\linewidth]{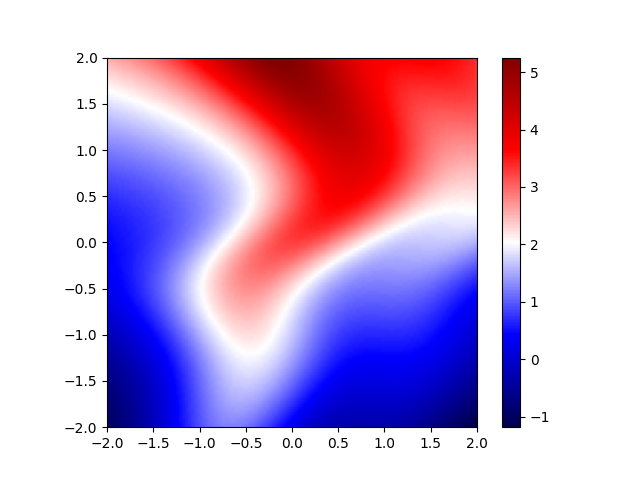}
    \caption{$\mtSigma_k \sim\mathcal{U}(0,2)$}
\end{subfigure}
\hfill
\begin{subfigure}{0.40\textwidth}
    \centering
    \includegraphics[width=\linewidth]{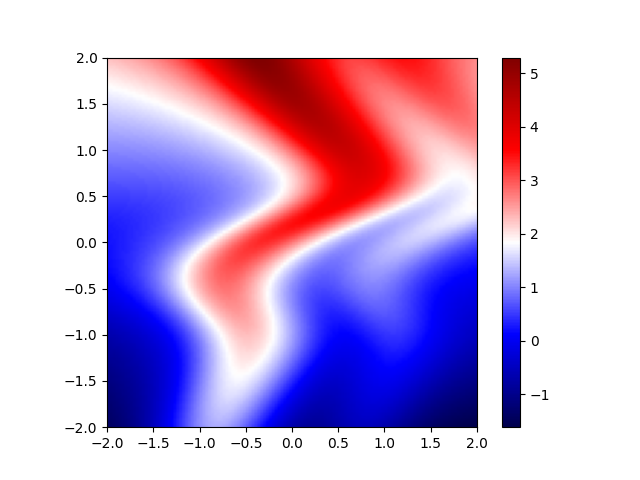}
    \caption{$\mtSigma_k \sim\mathcal{U}(0,4)$}
    \label{fig:tough-diffeo}
\end{subfigure}

\caption{A linear and representative nominal linearizing outputs for different scales of $\Sigma_k$ and fixed $W_k$, $c_k$ for the synthetic systems.}
\label{fig:synthetic-h}
\end{figure}

\begin{table}
\centering
\begin{tabular}{lcccc}
\toprule
 & Num. Kernels & $\Sigma_k$ & $c_k$ & $W_k$ \\
\midrule
$t_1$ & 32 & $\mathcal{U}(0,2)$ & $\mathcal{U}(-1,1)$ & $\mathcal{U}(-1,1)$ \\
$t_2$ & 32 & $\mathcal{U}(0,2)$ & $\mathcal{U}(-1,1)$ & $\mathcal{U}(-1,1)$ \\
\bottomrule
\end{tabular}
\caption{Kernel configuration for $t_1(x)$ and $t_2(x)$ in the synthetic 2D systems of the main paper.}
\label{tab:kernels1d}
\end{table}

\begin{table}
\centering
\begin{tabular}{lcccc}
\toprule
 & Num. Kernels & $\Sigma_k$ & $c_k$ & $W_k$ \\
\midrule
$t_1$ & 32 & $\mathcal{U}(0,2)$ & $\mathcal{U}(-1,1)$ & $\mathcal{U}(-1,1)$ \\
$t_2$ & 32 & $\mathcal{U}(0,1)$ & $\mathcal{U}(-0.8,0.8)$ & $\mathcal{U}(-1,1)$ \\
\bottomrule
\end{tabular}
\caption{Kernel configuration for $t_1(x)$ and $t_2(x)$ in the synthetic 4D systems of the main paper.}
\label{tab:kernels2d}
\end{table}

\begin{table}
\centering
\begin{tabular}{lcccc}
\toprule
 & Num. Kernels & $\Sigma_k$ & $c_k$ & $W_k$ \\
\midrule
$t_1$ & 16 & $\mathcal{U}(0,0.6)$ & $\mathcal{U}(-1,1)$ & $\mathcal{U}(-1,1)$ \\
$t_2$ & 16 & $\mathcal{U}(0,0.2)$ & $\mathcal{U}(-0.6,0.6)$ & $\mathcal{U}(-1,1)$ \\
\bottomrule
\end{tabular}
\caption{Kernel configuration for $t_1(x)$ and $t_2(x)$ in the synthetic 8D systems of the main paper.}
\label{tab:kernels8d}
\end{table}
\begin{figure}
    \centering
    \includegraphics[width=0.5\linewidth]{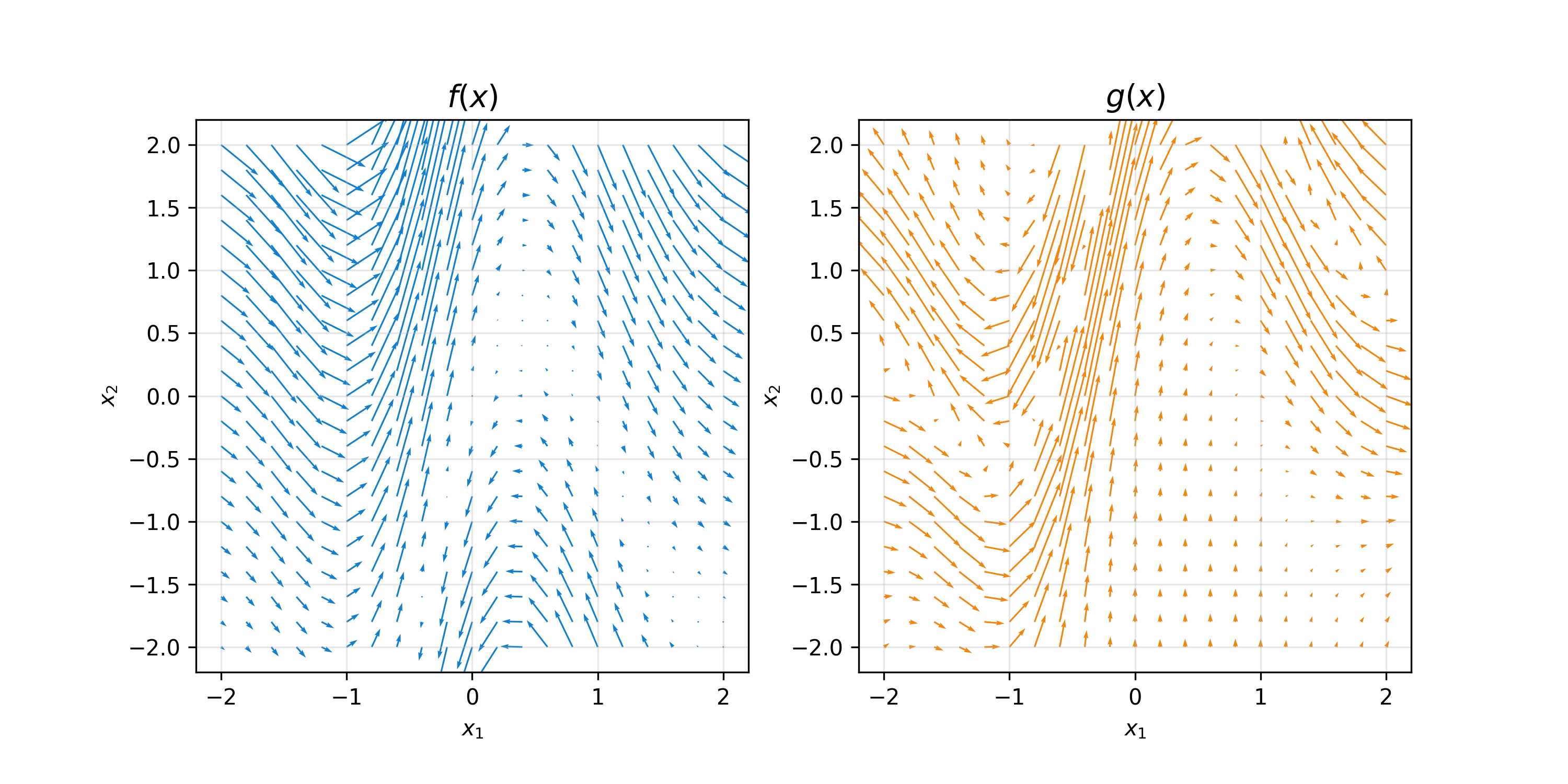}
    \caption{Sample vector fields $f$, $g$ of the 2D SISO synthetic systems generated from the configuration of Table~\ref{tab:kernels1d}.}
    \label{fig:synthetic-sys-viz}
\end{figure}
To demonstrate the applicability of the proposed method to control affine systems that are characterized by highly non-regular vector fields $f,g$ and are feedback linearized from less smooth diffeomorphisms as the one in Figure~\ref{fig:tough-diffeo}, we additionally provide numerical results for the configuration in Table~\ref{tab:kernels1dtough}.
\begin{table}
\centering
\begin{tabular}{lcccc}
\toprule
 & Num. Kernels & $\Sigma_k$ & $c_k$ & $W_k$ \\
\midrule
$t_1$ & 32 & $\mathcal{U}(0,4)$ & $\mathcal{U}(-1,1)$ & $\mathcal{U}(-1,1)$ \\
$t_2$ & 32 & $\mathcal{U}(0,4)$ & $\mathcal{U}(-1,1)$ & $\mathcal{U}(-1,1)$ \\
\bottomrule
\end{tabular}
\caption{Kernel configuration for $t_1(x)$ and $t_2(x)$ in the synthetic non-smooth 2D systems.}
\label{tab:kernels1dtough}
\end{table}
Representative vector fields are presented in Figure~\ref{fig:synthetic-sys-tough-viz}, and, following the evaluation methodology of the main paper, we provide the numerical results in Table~\ref{tab:2d-tough}. We observe that, even in this more challenging setup, only one out of ten systems exhibits unstable behavior, despite the absence of dedicated hyperparameter tuning during training.
\begin{figure}
    \centering
    \includegraphics[width=0.5\linewidth]{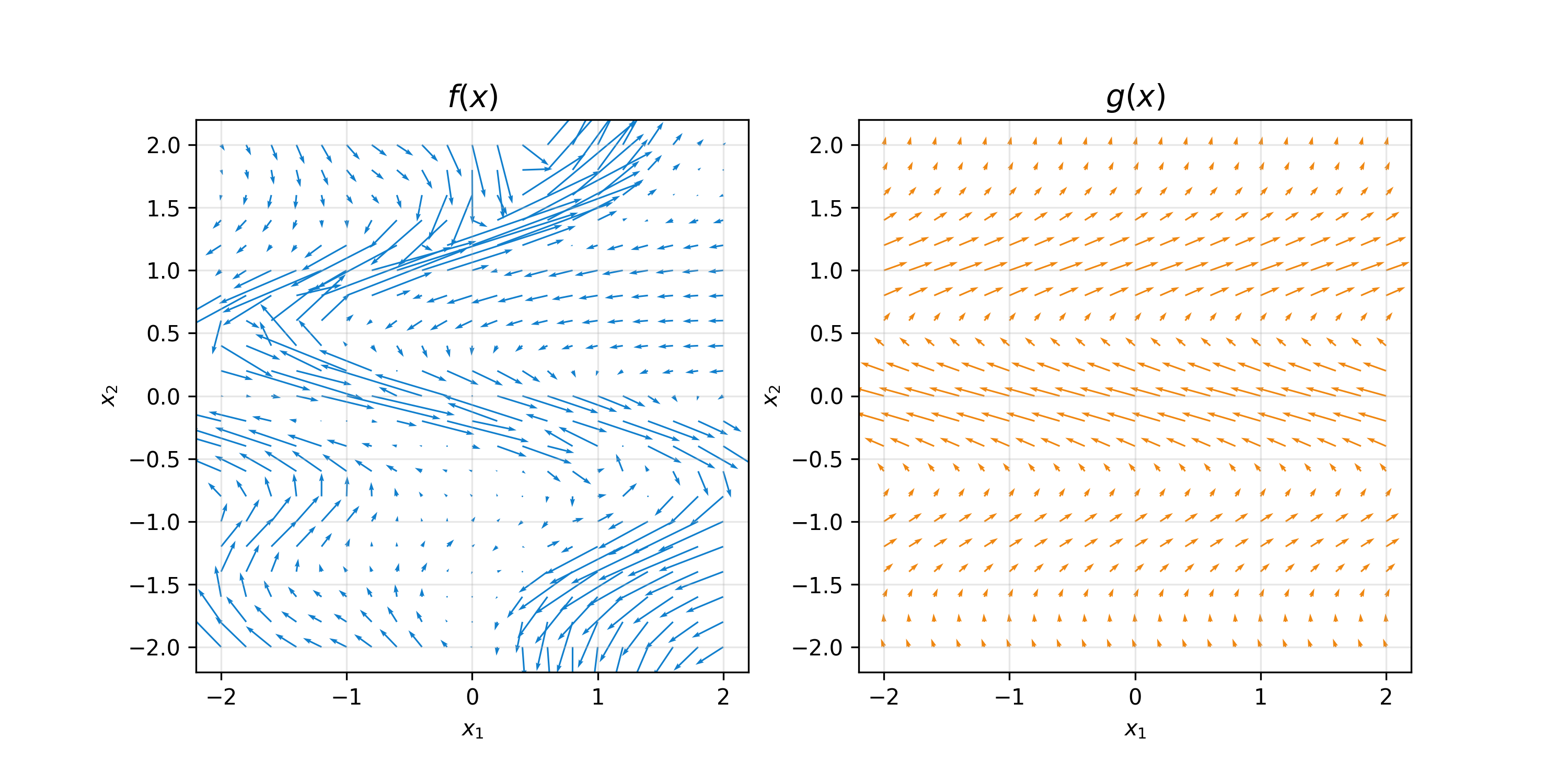}
    \caption{Sample vector fields $f$, $g$ of the 2D SISO synthetic systems generated from the configuration of Table~\ref{tab:kernels1dtough}.}
    \label{fig:synthetic-sys-tough-viz}
\end{figure}
\begin{table}[h]
\centering
\begin{tabular}{lcc}
\toprule
\textbf{Metric} & \textbf{Ground Truth} &\textbf{ Learned}\\
\midrule
RLQE& --- & $\mathbf{2.86 \times 10^{-4}}$ \\
\midrule
Avg. error & 0.1177& 0.1178 \\
%Avg. $z$-error & 0.0464 & 0.0472 \\
Final error & 0.0002 & 0.0002 \\
%Final $z$-error & 0.0001 & 0.0001 \\
Control effort & 42.70 & 42.74 \\
\midrule
Diverged (any) & --- & 0.1000 \\
Perc. diverged & --- & 0.00156 \\
\bottomrule
\end{tabular}
\caption{Tracking performance in 2D SISO systems generated by from the configuration of Table~\ref{tab:kernels1dtough}.}
\label{tab:2d-tough}
\end{table}